\documentclass[12pt]{article}
\usepackage{amsmath}
\usepackage{graphicx}
\usepackage{enumerate}
\usepackage{url} 

\usepackage[dvipsnames]{xcolor}

\newcommand{\blind}{1}

\usepackage[english]{babel}
\usepackage[utf8]{inputenc}
\usepackage[T1]{fontenc}
\usepackage{lmodern}
\usepackage{amssymb}
\usepackage{amsthm}
\usepackage{mathtools}
\usepackage{float}
\usepackage{bm}

\graphicspath{{images/}}
\usepackage[font=small,format=hang]{caption}
\usepackage{subcaption}
\usepackage{float}

\newtheorem{proposition}{Proposition}

\newtheorem{condition}{Conditions List}
\newtheorem{theorem}{Theorem}
\newtheorem{definition}{Definition}

\newcommand{\norm}[1]{\left\lVert#1\right\rVert}
\DeclareMathOperator{\tr}{tr}

\usepackage[authoryear]{natbib}
\setcitestyle{citesep={;}}
\setcitestyle{aysep={}} 

\usepackage{geometry}

\usepackage[linktoc=all,colorlinks=true,linkcolor=black,citecolor=blue,urlcolor=blue,bookmarks=false,hypertexnames=true]{hyperref}

\begin{document}

\setlength{\bibsep}{6.0pt}

\def\spacingset#1{\renewcommand{\baselinestretch}%
{#1}\small\normalsize} \spacingset{1}

\if1\blind
{
  \title{\Large\bf Doubly Robust Feature Selection with Mean and Variance Outlier Detection and Oracle Properties}
      \author{Luca Insolia\thanks{
            This work was partially funded by the Huck Institutes of the Life Sciences of Penn State.
            }
        \hspace{.2cm}\\
        \normalsize Faculty of Sciences, Scuola Normale Superiore \\
        \normalsize Institute of Economics \& EMbeDS, Sant'Anna School of Advanced Studies 
        \vspace{.2cm} \\
        Francesca Chiaromonte\\
        \normalsize Department of Statistics, Penn State University \\
        \normalsize Institute of Economics \& EMbeDS, Sant'Anna School of Advanced Studies 
        \vspace{.2cm} \\
        Runze Li\\
        \normalsize Department of Statistics, Penn State University
        \vspace{.2cm} \\
        Marco Riani\\
        \normalsize Department of Economics and Management, University of Parma
        }
  \date{}
} \fi

\if0\blind
{
  \bigskip
  \bigskip
  \bigskip
  \begin{center}
    {\LARGE\bf { 
    Doubly Robust Feature Selection with Mean and     Variance Outlier Detection and Oracle Properties
    } \\ 
    }
\end{center}
} \fi

\newgeometry{left=3cm,right=3cm, bottom=3cm, top=3cm}

\maketitle
\begin{abstract}

\noindent
We propose a general approach to handle data contaminations that might disrupt the performance of feature selection and estimation procedures for high-dimensional linear models.  
Specifically, we consider the co-occurrence of mean-shift and variance-inflation outliers, which can be modeled as additional fixed and random components, respectively, and evaluated independently. 
Our proposal performs feature selection while detecting and down-weighting variance-inflation outliers, detecting and excluding mean-shift outliers, and retaining non-outlying cases with full weights. 
Feature selection and mean-shift outlier detection are performed through a robust class of nonconcave penalization methods. 
Variance-inflation outlier detection is based on the penalization of the restricted posterior mode.
The resulting approach satisfies a robust oracle property for feature selection in the presence of data contamination -- which allows the number of features to exponentially increase with the sample size -- and detects truly outlying cases of each type with asymptotic probability one. 
This provides an optimal trade-off between a high breakdown point and efficiency.
Computationally efficient heuristic procedures are also presented.
We illustrate the finite-sample performance of our proposal through an extensive simulation study and a real-world application.
        
\end{abstract}

\noindent%
{\it Keywords:}
Mean-shift outliers;
Nonconvex penalties;
Robust estimation; 
Variable selection;
Variance-inflation outliers.

	\restoregeometry
	\section{Introduction}
	    
	    Modern regression problems encompass an ever increasing number of predictor variables, or features -- which motivates the use of feature selection techniques.
	    In the real-world, these problems are often also affected by data contamination, e.g.,~due to recording errors or the presence of different sub-populations. 
	    Handling the resulting outliers is critical, as data contamination can hinder classical feature selection and estimation methods. 
	    Moreover, outlier detection itself can be a major goal of the analysis, as it often provides valuable domain-specific insights.

        Two main contamination mechanisms have been investigated in the literature on linear models \citep{beckman1983outlier}, namely: the \textit{mean-shift outlier model} (MSOM) and the \textit{variance-inflation outlier model} (VIOM).
        The MSOM assumes that outlying cases have a shift in mean; \textit{maximum likelihood estimation} (MLE) leads to their removal from the fit -- i.e.,~to the assignment of $0$ weights to the cases identified as outliers. 
        While the MSOM was traditionally studied in low-dimensional scenarios \citep{cook1982residuals}, it has been recently extended to high-dimensional linear models, where the use of regularization techniques is fundamental \citep{she2011outlier,alfons2013sparse,kurnaz2017robust,insolia2020simultaneous}.
        The VIOM, which is historically considered as an alternative to the MSOM, assumes that contaminated errors have an inflated variance; outliers are retained but down-weighted in the fit. 
        The VIOM was initially investigated by \citet{cook1982note} and \citet{thompson1985note} in the presence of a single outlier, using MLE and \textit{restricted MLE} (REMLE), respectively.
        More recently, \citet{gumedze2019use} developed hypothesis testing procedures for linear models, considering also the presence of multiple outliers. 
        However, when multiple outliers are present, this approach requires the evaluation of a combinatorial number of outlying-ness tests to avoid masking (undetected outlying cases) and swamping (non-outlying cases flagged as outliers).
        \citet{insolia2020ViomMsom} proposed the use of robust estimation and REMLE to detect and down-weight multiple VIOM outliers, possibly co-occurring with MSOM outliers, 
        in (low-dimensional) linear models.

        High-dimensional settings with VIOM outliers, to the best of our knowledge, have not been explored yet. 
        Here we aim to fill this gap and, like in \citet{insolia2020ViomMsom}, we further consider the co-occurrence of multiple MSOM and VIOM ouliers. 
        These are modeled as additional fixed and random components, respectively, which can be estimated independently based on REMLE principles.
        Specifically, we propose a doubly robust class of nonconcave penalization methods, in which feature selection and MSOM detection rely on a trimmed penalized loss, whereas VIOM detection is based on the penalization of the restricted posterior mode.
        The resulting procedure: 
        (i) satisfies a robust oracle property for feature selection in the presence of data contamination, which allows the number of features to exponentially increase with the sample size;
        (ii) detects MSOM and VIOM outliers with asymptotic probability one;
        (iii) achieves an optimal trade-off between high breakdown point and efficiency, and thus provides optimal units' weights.
        Effective and computationally efficient heuristic procedures are also presented.

        Importantly, our approach comprises ``hard'' trimming sparse estimators as a special case. However, since we rely on nonconcave penalization methods, our proposal satisfies oracle properties under weaker assumptions compared to existing robust estimators based on convex penalties  \citep{kurnaz2017robust,alfons2013sparse}.
        This provides an important bridge between the latter and $L_0$-constrained formulations with optimality guarantees \citep{insolia2020simultaneous}.
        Moreover, unlike ``soft'' trimming estimators which produce a general down-weighting for all points \citep{,loh2017statistical,smucler2017robust,chang2018robust,freue2019robust,amato2021penalised}, our proposal is effective in estimating full weights for non-outlying observations.
        
        The reminder of the paper is organized as follows. 
        Section~\ref{sec:background} reviews relevant background literature.        
        Section~\ref{sec:proposal} details our proposal, which is a 3-step procedure, as well as its heuristic counterpart.
        Section~\ref{sec:sim} contains numerical studies comparing the empirical properties of different methods both in low- and high-dimensional settings, and Section~\ref{sec:appl} contains a real-world application.
        Final remarks are given in Section~\ref{sec:final}. 
        Further details, extensions and proofs,
        as well as the source code to replicate our simulation and application studies, are provided in the Supplementary Material.

    \section{Background} 
    \label{sec:background}   
    
        In this section we review two streams of literature that are relevant for our developments; namely, methods for outlier detection in low-dimensional linear models, and approaches for feature selection in high-dimensional mixed-effects linear models.

    \subsection{Outlier Detection} 
    \label{secsub:outlier_detection}
    
    	Consider a classical linear regression model of the form
    	$
        	\bm{y} = \bm{X} \bm{\beta} + \bm{\varepsilon} , 
    	$
    	where $\bm{y} = ( y_1, \ldots, y_n )^T \in \mathbb{R}^n $ contains observable responses, 
    	$ \bm{X} = ( \bm{x}_1, \ldots, \bm{x}_n)^T  \in \mathbb{R}^{n \times p} $ is the design matrix, 
    	$ \bm{\beta} \in \mathbb{R}^p $ contains unknown fixed effects (possibly sparse), 
    	and $ \bm{\varepsilon}  = ( \varepsilon_1, \ldots, \varepsilon_n )^T \in \mathbb{R}^n $ contains unobservable random errors.
    	Classical assumptions specify that such errors are uncorrelated, homoscedastic and Gaussian, so that $ \bm{\varepsilon} \sim N(\bm{0}, \sigma^2 \bm{I}_n) $
    	for $0 < \sigma^2 < \infty$.
    
        The MSOM postulates that for outlying cases $i \in \mathcal{S}_\phi$ (the rationale for this symbol will become clear in Equation~\ref{eq:equivalLMMmsomViom}), $\varepsilon_i \sim N(\mu_{\varepsilon_i}, \sigma^2)$ with $\mu_{\varepsilon_i} \neq 0$.
        Under the assumption that $\mathcal{S}_\phi$ is known
        and  $ \text{rank}(\bm{X})=p \leq n - \lvert \mathcal{S}_\phi \rvert  $ (where $ \lvert \cdot \rvert $ denotes the cardinality of a set), the MLE leads to the exclusion of the units in $\mathcal{S}_\phi$ from the fit \citep{cook1982residuals}.
        If there is a single MSOM outlier, this represents the unit with largest absolute Studentized residual, which is a monotone transformation of the deletion residual 
        $
    	    t_i = (y_i- \bm{x}_i^{T} \widehat{\bm{\beta}}_{(i)}) / \{ \widehat{\sigma}_{(i)} 
    	    (1+ \bm{x}_{i}^{T}( \bm{X}_{(i)}^{T} \bm{X}_{(i)})^{-1} \bm{x}_i)^{1 / 2} \},
	    $
        where the parenthetical subscript indicates the exclusion of unit $i$ from the fit. 
        Importantly, $t_i$ can be computed very cheaply and, for a generic $i$, follows a Student's $t$ with $n-p-1$ degrees of freedom under the null -- thus, it can be used to test the outlying-ness of each observation. 
        Although this can be easily generalized to the presence of multiple MSOM outliers, it requires the evaluation of a combinatorial number of fits (i.e.,~excluding all possible subsets of points of a given size from the fit), which results in a computationally intractable
        problem. 
        Relatedly, high-breakdown estimators (see Section~\ref{secsub:proposal_fixed}) aim at limiting the influence of extreme residuals on the fit  \citep{maronna2006robust}.
        Although these are traditionally computed using heuristic approaches, the use of MIP techniques has been recently considered to effectively solve the underlying combinatorial problem with optimality guarantees \citep{zioutas2005deleting,bertsimas2014least}.
        Importantly, high-breakdown point estimators have also been extended to sparse high-dimensional linear models in combination with penalization methods         \citep{alfons2013sparse,smucler2017robust,kurnaz2017robust,freue2019robust}. Here $L_0$-constraints, which can be solved through MIP algorithms, provide optimality guarantees and desirable statistical properties for simultaneous feature selection and MSOM detection, with $p$ allowed to increase exponentially with $n$ \citep{insolia2020simultaneous}.

        The VIOM postulates that for outlying cases $i \in \mathcal{S}_\gamma$ (also this symbol will become clear in Equation~\ref{eq:equivalLMMmsomViom}), $\varepsilon_i \sim N(0, \sigma^2 v_i) $ with $v_i= (1+ \omega_i) \geq 1$.
    	\citet{cook1982note} studied the presence of a single variance-inflated outlier; the MLE estimate of $\bm{\beta}$ depends on its $v_i$ and results in a \textit{weighted least squares} (WLS) fit 
		$
			\widehat{\bm{\beta}}(v_i) =( \bm{X}^T \bm{W} \bm{X})^{-1} \bm{X}^T \bm{W} \bm{y}
			= \widetilde{\bm{\beta}} - (\bm{X}^T \bm{X})^{-1} \bm{X}_{i}^T \widetilde{e}_{i} [ (1-w_i) / \{ 1-(1-w_i)H_{x, ii} \} ] ,
		$
		where $\bm{W}$ is a diagonal matrix containing all ones but $w_i = v_i^{-1}$.
		The tilde indicates quantities computed from the \textit{ordinary least squares} (OLS) fit, and $ H_{x, ii} $ is the $i$-th diagonal element of $\bm{H}_x = \bm{X} (\bm{X}^T \bm{X})^{-1} \bm{X}^T $.
    	This highlights the fact that the VIOM is asymptotically equivalent to the MSOM as $v_i \to \infty$.
        Importantly, in the presence of a single VIOM outlier, the MLE provides a closed-form estimate for $v_i$, which can be used to estimate $\bm{\beta}$ and $\sigma^2$.
    	Similarly, \citet{thompson1985note} used REMLE in place of MLE to estimate the variance components $ v_i $ and $\sigma^2$.
        REMLE relies on $n - p$ linearly independent error contrasts $\bm{A}^T \bm{\varepsilon}$, where $\bm{A} \in \mathbb{R}^{n \times (n-p)}$ is defined such that $ \bm{A}^T \bm{A} = \bm{I}_n$ and $\bm{A} \bm{A}^T = \bm{P}_x$, with $\bm{P}_x = \bm{I}_n - \bm{H}_x $ \citep{patterson1971recovery}.
        Also REMLE provides a closed-form estimate for the single variance-inflation parameter $v_i$.
        Notably, the single VIOM outlier position estimated by MLE and REMLE might differ.
    	A sufficient condition for their agreement is that the unit with maximum absolute OLS residual $ \max( \lvert \widetilde{e}_i \rvert) $ also has the largest absolute Studentized residual $ \max( \lvert t_i \rvert )$ -- the latter estimates the outlier position using REMLE, which is equivalent to the outlier position estimated by MLE under an MSOM \citep{thompson1985note}.
    	However, differently from the case of a single VIOM outlier (and of multiple MSOM outliers), multiple variance-inflation parameters $\bm{v}$ cannot be estimated in closed-form even if the outliers are known -- thus, iterative procedures are required \citep{gumedze2019use}.
        In order to detect multiple VIOM outliers, possibly concurrent with MSOM outliers, \citet{insolia2020ViomMsom} proposed the use of robust estimation for outlier detection and of REMLE to estimate optimal units' weights.
        Nevertheless, to the best of our knowledge, high-dimensional linear models affected by VIOM contamination have not been explored yet.

    \subsection{Feature Selection for Mixed-Effects Linear Models} \label{secsub:mixed_model}
        
        Mixed-effects linear models are often used to model data with a natural group structure, such as repeated measurements, measurements in time, and measurements in space \citep{laird1982random}. 
        They extend the classical linear model through the inclusion of a random design matrix characterizing the experiment; namely,        
        $
        	\bm{y} = \bm{X} \bm{\beta} + \bm{Z} \bm{b} + \bm{\varepsilon} ,
        $
    	where $ \bm{Z} = [\bm{Z}_1, \ldots, \bm{Z}_t] \in \mathbb{R}^{n \times q} $, 
    	and $\bm{Z}_j \in \mathbb{R}^{n \times q_j}$ indicates  the design matrix for the $j$-th random effect
    	$\bm{b}_j \in \mathbb{R}^{q_j}$, such that $ \bm{b} = ( \bm{b}_1^T, \ldots, \bm{b}_t^T )^T \in \mathbb{R}^{q} $,
    	and $ \sum_j q_j = q $. 
    	It is often assumed that $ \bm{b} \sim N ( \bm{0}, \bm{\mathcal{B}})$, where 
    	$\bm{\mathcal{B}} = [ \bm{B}_1, \ldots, \bm{B}_t ] $ is a block-diagonal matrix modeling the covariance of each random effect $\bm{b}_j\sim N ( \bm{0}, \bm{B}_j)$, with $\text{cov}(\bm{b}_k, \bm{b}_l) = 0 $ for any $ k \neq l$.
    	Moreover, $ \bm{b} $ and $\bm{\varepsilon} $ are assumed to follow independent Gaussian distributions.
            
        Several methods have been developed to simultaneously estimate fixed and random effects.
        Henderson's mixed-model equations lead to the \textit{best linear unbiased estimator} (BLUE) for the fixed effects $\bm{\beta}$ and the \textit{best linear unbiased predictor} (BLUP) for the random effects $ \bm{b} $ 
        -- which is also known as the \textit{empirical Bayes estimator} as it maximizes the posterior distribution $f( \bm{b} \lvert \bm{y} ) $. 
        However, this approach is unviable to perform feature selection in high-dimensional scenarios \citep{fan2012variable}. 
        For this purpose, hypothesis testing procedures have been developed to select relevant random effects \citep{lin1997variance}.
        Different sub-models can be compared through extensions of information criteria, such as the \textit{conditional Akaike information criterion} (CAIC) \citep{liang2008note} and its generalizations.
        Leveraging penalization methods, other approaches perform sparse estimation of the fixed effects $\bm{\beta}$. 
        In these, while the dimension $p$ of $\bm{\beta}$ is allowed to increase with the sample size $n$, the random component $\bm{b}$ is often assumed to contain only truly  relevant random effects \citep{schelldorfer2011estimation}.
        Yet other approaches use penalization methods to select a given number of fixed and random effects \citep{bondell2010joint,ibrahim2011fixed,peng2012model}. 
        See \citet{muller2013model} and \citet{buscemi2019model} for a literature review.

        In the following we focus on the class of nonconcave penalization methods introduced by \citet{fan2012variable}.
        Importantly, based on REMLE principles, selection of fixed and random effects can be performed independently.
        Under mild conditions this approach satisfies a weak oracle property for fixed effects estimates and selects truly relevant random effects with asymptotic probability one -- where the dimensions $p$ and $q$ of fixed and random effects are allowed to exponentially increase with the sample size.

    \section{Our Proposal} 
    \label{sec:proposal} 
        
        We investigate linear models affected by systematic (MSOM) and/or stochastic (VIOM)  contaminations. 
        Specifically, we focus on a general \textit{unlabeled} outlier problem \citep{beckman1983outlier}, where the nature (MSOM vs.~VIOM) as well as the identity, number and strength of the outliers is unknown.
        We model the presence of $m_V$ VIOM and $m_M$ MSOM outliers,  indexed through the  (unknown and non-overlapping) sets $\mathcal{S}_\gamma$ and $\mathcal{S}_\phi$: 
    	\begin{equation}
        	\varepsilon_i \sim 
        	\begin{cases}
            	N(0, \sigma^2 v_i) \quad &\text{$\forall ~ i \in \mathcal{S}_\gamma$} \\
            	N(\mu_{\varepsilon_i}, \sigma^2) \quad &\text{$\forall ~ i \in \mathcal{S}_\phi$} \\
            	N(0, \sigma^2) \quad &\text{otherwise} , \label{eq:cont_model}
        	\end{cases}
    	\end{equation}
    	where $v_i >1$ and $\mu_{\varepsilon_i} \neq 0$.
    	We exclude overlaps between the two types of contamination because such over-parametrization is equivalent to a MSOM assumption \citep{cook1982note}.
        Moreover, as customary in the robust statistics literature, we let MSOM outliers also affect the design matrix $\bm{X}$ (with shifts $\mu_{x_i}$) creating leverage points \citep{maronna2006robust}.

    	Notably, the outliers in \eqref{eq:cont_model} can be equivalently represented adding fixed and random effects to the linear model \citep{insolia2020ViomMsom}. In symbols
    	\begin{equation} \label{eq:equivalLMMmsomViom}
        	\bm{y} = \bm{X} \bm{\beta} + \bm{D}_{\mathcal{S}_\gamma} \bm{\gamma} + \bm{D}_{\mathcal{S}_\phi} \bm{\phi} + \bm{\epsilon} ,
    	\end{equation}
    	where $\bm{D}_{\mathcal{S}_\gamma} $ ($n \times m_V$) and $\bm{D}_{\mathcal{S}_\phi} $ ($n \times m_M$)  are matrices composed by dummy column vectors indexing VIOM and MSOM outliers, respectively.
    	The $m_V \times 1$ random vector $\bm{\gamma} \sim N(\bm{0}, \sigma^2 \bm{\Gamma})$ allows one to down-weight VIOM outliers; here $ \bm{\Gamma}= \text{diag}_{m_V}( \bm{\omega} )$ is a diagonal matrix of size $m_V$.
        The non-stochastic vector 
    	$\bm{\phi} \in \mathbb{R}^{m_M} $ contains prediction residuals for MSOM outliers (i.e.,~their residuals based on an estimator which excludes them from the estimation process) and removes their influence from the fit.
    	The associated $t$-statistics are the deletion residuals  $t_{\mathcal{S_\phi}}$.
        The random error vector is assumed to be $\bm{\epsilon} \sim N(\bm{0}, \sigma^2 \bm{I}_n)$ and independent from $\bm{\gamma}$.
        If the sets of outliers $\mathcal{S}_\phi$ and $\mathcal{S}_\gamma$ are known, and $ \text{rank}(\bm{X})=p \leq n - m_M $, the formulation in \eqref{eq:equivalLMMmsomViom} allows one to use standard techniques for mixed-effects linear models to estimate variance-inflation parameters $\bm{v}$ and regression coefficients $\bm{\beta}$. 
        However, this approach is unfeasible if the outlier identities are unknown and/or if $p>n$.
        To tackle this problem, we consider the general formulation
    	\begin{align} \label{eq:equivalLMMmsomViom2}
        	\bm{y} &= 
        	\bm{X} \bm{\beta} + \bm{I}_n \bm{\gamma}  + \bm{I}_n \bm{\phi} + \bm{\epsilon}
    	\end{align}
         and rely on nonconcave penalization methods to select relevant fixed effects $\bm{\beta}$ -- but we also enforce sparsity in $\bm{\gamma} \in \mathbb{R}^{n}$, which detects and down-weights VIOM outliers, and $\bm{\phi}  \in \mathbb{R}^{n}$, which detects and excludes MSOM outliers from the fit.
         Specifically, we propose a 3-step procedure based on REMLE principles, that extends and combines the approaches in \citet{fan2012variable} and  \citet{insolia2020simultaneous,insolia2020ViomMsom}.
         Operationally, the three steps can be solved iteratively (see Section~\ref{sec:sim}), and we first focus on fixed effects estimation, as MSOM outliers can have stronger influence on model estimates.
        
    \subsection{Step 1: Feature Selection and MSOM Detection}
    \label{secsub:proposal_fixed}

        Suppose that $\mathcal{S}_\gamma$ is known. 
        Then, plugging the MLE estimates for $\bm{\gamma} \lvert \bm{\beta}$ in the joint density distribution $f(\bm{y},\bm{\gamma})$ leads to the profile log-likelihood:
        \begin{align} \label{eq:loglik_fix}
            l_n (\bm{\beta},\bm{\widehat{\gamma}}) \propto 
            \frac{1}{2 \sigma^2} 
            (\bm{y} - \bm{X} \bm{\beta} - \bm{\phi} )^T  \bm{P}_R
            (\bm{y} - \bm{X} \bm{\beta} - \bm{\phi} ) ,
        \end{align}
        which produces a WLS estimator as
        \begin{align}
            \bm{P}_R 
             = ( \bm{I}_n & - \bm{B}_\gamma)^T ( \bm{I}_n - \bm{B}_\gamma) + \bm{B}_\gamma^T
             \bm{D}_{\mathcal{S}_\gamma}     
             \bm{\Gamma}^{-1} 
             \bm{D}_{\mathcal{S}_\gamma}^T
             \bm{B}_\gamma \nonumber  \\
                 = (\bm{I}_n & +  \bm{D}_{\mathcal{S}_\gamma} \bm{\Gamma}
                 \bm{D}_{\mathcal{S}_\gamma}^T
                 )^{-1} = \bm{W}, 
                 \label{eq:Pgamma}  
        \end{align}        
        where $\bm{B}_\gamma = (\bm{I}_n + 
                \bm{D}_{\mathcal{S}_\gamma} 
                \bm{\Gamma}^{-1}
                \bm{D}_{\mathcal{S}_\gamma}^T
                )^{-1} $.
        We simultaneously select and estimate fixed effects $\bm{\beta}$, while detecting and discarding MSOM outliers from the fit, using a feasible and robustly penalized version of \eqref{eq:loglik_fix}, where an integer constraint and a nonconcave penalty are used for MSOM outlier detection and feature selection, respectively.
        In symbols
        \begin{subequations}\label{eq:reg_fixVIOM_MSOM}
	       \begin{align} 
                \left[ \widehat{\bm{\beta}}, \widehat{\bm{\phi}} \right] = \operatorname*{\arg\min}_{\bm{\beta}, \bm{\phi}} ~& \frac{1}{2} 
                (\bm{y} - \bm{X} \bm{\beta} - \bm{\phi} )^T  \bm{\mathcal{M}}_R
                (\bm{y} - \bm{X} \bm{\beta} - \bm{\phi}) + 
                (n - k_n) \sum_{j=1}^{p} R_{\lambda}(\lvert \beta_j \lvert) 
                \tag{\ref{eq:reg_fixVIOM_MSOM}}
                \\
                \text{s.t.} ~
                & \norm{\bm{\phi}}_0 = \sum_{i=1}^n I(\phi_i \neq 0) \leq k_n  , \label{eq:contrL0}
            \end{align}
        \end{subequations}
        where $  I( \cdot ) $ is the indicator function, and the matrix $\bm{\mathcal{M}}_R$ is a proxy for the unknown $\bm{P}_R/\sigma^2$ (see the Supplementary Material for details).
        Note that if $\bm{\mathcal{M}}_R$ is a multiple of the identity matrix, then \eqref{eq:reg_fixVIOM_MSOM} neglects VIOM outliers -- i.e.,~all points receive binary weights.
                
        The penalty function $ R_{\lambda}(\cdot) $ enforces sparsity in $\bm{\beta}$ estimates and depends on a tuning parameter $\lambda $ controlling the trade-off between goodness of fit and model complexity.
        For this task, several penalties have been investigated in the literature. \citet{tibshirani1996regression} introduced the \textit{lasso} based on the $L_1$-penalty, which is very efficient but provides biased estimates.
        To overcome this limitation, nonconcave penalties have also been used. These include the \textit{smoothly clipped absolute deviation} (SCAD) \citep{fan2001variable}, the \textit{minimax concave penalty} (MCP) \citep{zhang2010nearly}, and the \textit{adaptive lasso} \citep{zou2006adaptive}.
        Other approaches solve the combinatorial best subset selection problem using an $L_0$-constraint and MIP algorithms. 
        In this work we focus on penalties satisfying the following conditions.
        \begin{condition}[Penalty function] \label{cond1}
            For any $\lambda > 0$, the penalty $R_{\lambda}(t)$, $t \in [0, \infty )$ is: 
            (i)
            non-decreasing and concave with $ R_\lambda(0) = 0 $,
            (ii)
            twice continuously differentiable with first derivative  $R_\lambda'(0^+) > 0 $, and
            (iii)
            such that
            $\sup_{t>0} R_\lambda''(t) \to  0$ for $\lambda \to 0$.
        \end{condition}
        \noindent 
        These conditions are fairly common for concave penalization methods (see for instance \citealt{fan2011nonconcave}), and are used to develop estimators with three desirable properties: unbiasedness, sparsity and continuity \citep{fan2001variable}. 
        We specifically focus on the SCAD penalty $R_\lambda(\cdot)$ in \eqref{eq:reg_fixVIOM_MSOM}, but others might be considered as well.
        The SCAD penalty satisfies $R_\lambda( 0 ) = 0$ and, for $t \in (0, \infty )$, has $R_\lambda'( t ) = \lambda I( t \leq \lambda) + [ (a \lambda - t) / (a -1) ] I (t > \lambda )$, where the constant $a>2$ controls nonconcavity and is often set to $a=3.7$.
        This folded-concave penalty is continuously differentiable on $ (-\infty,0) \cup (0, \infty)$ and singular at 0. Since its derivative is zero outside $ [ - a \lambda , a \lambda ]$, it does not shrink and thus bias large coefficient estimates. 
        Obtaining a global minimum with folded-concave penalties such as SCAD is non-trivial. In the following we focus on the \textit{local linear approximation} (LLA) method \citep{zou2008one} to obtain a local solution which guarantees oracle properties. 
        However, in principle one can achieve the global minumum using MIP techniques \citep{liu2016global}.
        
        The $L_0$-constraint in \eqref{eq:contrL0} is used for MSOM outlier detection.
        It depends on an integer tuning parameter $ k_n \geq 0$ controlling the trimming level -- i.e.,~the number of points which are identified as MSOMs and excluded from the fit.
        This guarantees the achievability of high-breakdown estimates (see below).
        Modern MIP solvers can be used to solve the formulation in  \eqref{eq:reg_fixVIOM_MSOM} with optimality guarantees \citep{bertsimas2016best,insolia2020simultaneous,kenney2021efficient}.
        However, in order to reduce the computational burden, one can also use well-established heuristic algorithms \citep{alfons2013sparse,kurnaz2017robust}.
        
        Intuitively, the \textit{breakdown point} (BdP) measures the largest fraction of contamination that an estimator can tolerate before it becomes arbitrarily biased \citep{donoho1983notion}.
        The finite-sample replacement BdP is defined as
        $
            \varepsilon^*( \widehat{\bm{\beta}}, \bm{Z})  = \min (m / n : \sup_{\widetilde{\bm{Z}}} \lVert  \widehat{\bm{\beta}}(\widetilde{\bm{Z}}) \rVert_2  = \infty ),
        $
        where $ \widetilde{\bm{Z}} $ denotes the original dataset $ \bm{Z} = (\bm{X}, \bm{y}) $ after the replacement of $m$ out of $n$ points with arbitrary values.
        The following result shows that our proposal achieves the highest possible BdP.
        \begin{proposition}[High breakdown-point] \label{thm:bdp}
            For any $\lambda > 0$ and $a>2$ the estimator $\widehat{\bm{\beta}}$ produced by \eqref{eq:reg_fixVIOM_MSOM} achieves a breakdown point of $  \varepsilon^* = (k_n + 1)/n$.
        \end{proposition}
        \noindent 
        Thus, in the presence of MSOM contamination, our proposal breaks down only if $k_n < m_M $.
        Moreover, this result does not require that the points $(\bm{x}_i^T, y_i ) $ are in general position. This is necessary for low-dimensional estimators to achieve equivariance \citep{maronna2006robust} -- something that cannot be achieved by our proposal \citep{maronna2011robust}.
        
        Note that lasso estimation can be considered as the first iteration in computing the SCAD penalty based on the LLA method \citep{zou2008one}. Thus, while SCAD provides stronger theoretical results for feature selection, one can perform MSOM outlier detection with existing robust algorithms based on lasso, e.g.,~the \textit{sparseLTS} \citep{alfons2013sparse} which solves a trimmed loss problem with an $L_1$-penalty using heuristic algorithms. Then, SCAD can be computed on the set of non-outlying cases detected by a robust lasso on the first iteration of LLA; this is the approach followed in our implementation described below.
        
        We remark that the notion of breakdown can be misleading for non-equivariant estimators, such as those produced through penalties \citep{maronna2011robust,smucler2017robust,insolia2020simultaneous}. Hence, we provide additional guarantees in terms of simultaneous MSOM outlier detection and feature selection.
        Let $ \bm{\theta}_0 = ( \bm{\beta}^T_0 , \bm{\phi}^T_0 )^T \in \mathbb{R}^{p+n} $ be the true parameter vector, and decompose it as 
        $ 
            \bm{\theta}_0 =  
            ( \bm{\theta}_{\mathcal{S}}^T , \bm{\theta}_{\mathcal{S}^c}^T )^T = 
                \{ ( \bm{\beta}^T_{\mathcal{S}_\beta}  , \bm{\phi}^T_{\mathcal{S}_\phi} ) , \\
                ( \bm{\beta}^T_{\mathcal{S}^c_\beta} , \bm{\phi}^T_{\mathcal{S}^c_\phi} ) \}^T 
        $
        where $ \bm{\theta}_{\mathcal{S}}  $ contains the $ p_0 $ non-zero coefficients belonging to $\mathcal{S}_\beta$, and the $ m_M $ outlying cases belonging to $\mathcal{S}_\phi$ ($(\cdot)^c$ indicates the  complement of a set).
        $\widehat{\bm{\theta}}_0$ represents a \textit{fixed-effects robust oracle estimator}, behaving as if the true sets of active features and outliers were both known in advance.  
        Let $\|\cdot\|_{\infty}$ indicate the matrix infinity norm, and $\Lambda_{\text{min}}(\cdot)$ and $\Lambda_{\text{max}}(\cdot)$ the minimum and maximum eigenvalue of a matrix, respectively. 
        We rely on the following conditions to recover $\widehat{\bm{\theta}}_0$. 
        \begin{condition}[Fixed-effects robust oracle reconstruction] \label{cond2}
            \item[A.] \underline{Minimum signal strength}:
                $  s_1 n^{\tau} \{ \log (n - m_M) \}^{-3 / 2} \to \infty $, 
                where
                $s_1=\min _{ j \in \mathcal{S}_\beta}\left|\beta_{0, j}\right| $,
                $\tau \in (0, 1/2 )$ is a given constant, and $\sup _{t \geq s_1 / 2} R_{\lambda}^{\prime \prime}(t)=o\left( (n - m_M)^{-1+2 \tau}\right) .$ 
            \item[B.] \underline{Design and proxy matrices}: 
                for some constants $\eta \in(2 \tau, 1]$ and $c_{0}>0$, the matrices $(n-m_M)^{-1} (\mathbf{X}_{\mathcal{S}_\phi^c , \mathcal{S}_\beta}^{T} \mathbf{X}_{\mathcal{S}_\phi^c , \mathcal{S}_\beta} )$ and $(n-m_M)^{\eta} (\mathbf{X}_{\mathcal{S}_\phi^c , \mathcal{S}_\beta}^{T} \mathbf{P}_R \mathbf{X}_{\mathcal{S}_\phi^c , \mathcal{S}_\beta})^{-1}$ have minimum and maximum eigenvalues bounded from below and above by $c_{0}$ and $c_{0}^{-1} $, respectively.
            Moreover 
            $$
                \left\|\left(\frac{1}{n-m_M} \mathbf{X}_{\mathcal{S}_\phi^c , \mathcal{S}_\beta }^{T} \bm{\mathcal{M}}_{R} \mathbf{X}_{\mathcal{S}_\phi^c , \mathcal{S}_\beta}\right)^{-1}\right\|_{\infty} \leq 
                \frac{\{ \log ( n-m_M ) \}^{3 / 4}}
                {(n - m_M)^{\tau} R_{\lambda}^{\prime}\left(s_1 / 2\right) },
            $$
            $$
                \left\|\mathbf{X}_{\mathcal{S}_\phi^c , \mathcal{S}_\beta^c}^{T} \bm{\mathcal{M}}_{R} \mathbf{X}_{\mathcal{S}_\phi^c , \mathcal{S}_\beta}\left(\mathbf{X}_{\mathcal{S}_\phi^c , \mathcal{S}_\beta}^{T} \bm{\mathcal{M}}_{R} \mathbf{X}_{\mathcal{S}_\phi^c , \mathcal{S}_\beta}\right)^{-1}\right\|_{\infty} < 
                \frac{R_{\lambda}^{\prime}(0+)}
                {R_{\lambda}^{\prime}\left(s_1 / 2\right)}
                .
            $$
            \item[C.] \underline{Proxy matrix}:
                $\Lambda_{\min }\left(c_{1} \bm{\mathcal{M}}_\gamma^{\mathcal{S}_\gamma} -  \bm{\Gamma}  \right) \geq 0$ and $\Lambda_{\min }\left( c_{1}  \log( n - m_M ) \bm{\Gamma} - \bm{\mathcal{M}}_\gamma^{\mathcal{S}_\gamma}    \right) \geq 0$ for some constant $c_{1}>0$, and 
                $ 
                \bm{\mathcal{M}}_\gamma^{\mathcal{S}_\gamma^c} = \bm{I}_{n-m_V}
                $. 
                Here 
                $ \bm{\mathcal{M}}_\gamma^{\mathcal{S}_\gamma^c} $ and $ \bm{\mathcal{M}}_\gamma^{\mathcal{S}_\gamma} $ index rows and columns of the proxy matrix $\bm{\mathcal{M}}_\gamma$ 
                corresponding to non-VIOMs and VIOMs, respectively.
            \item[D.] \underline{MSOM strength}:
                $\Delta_{\phi} \geq  d_\phi \sigma^2 \log (n) / n $, where $d_\phi > 0$ is a constant independent of $n$ and $p$, and 
                $$
                    \Delta_{\phi} = \min_{\widehat{\bm{\phi}}_{\widetilde{\mathcal{S}}_\phi}, \widehat{\bm{\beta}}_{\widetilde{\mathcal{S}}_\beta}} 
                    \frac{\| \bm{X}_{\widetilde{\mathcal{S}}_\beta} \widehat{\bm{\beta}}_{\widetilde{\mathcal{S}}_\beta} + \bm{I}_{n, \widetilde{\mathcal{S}}_\phi}\widehat{\bm{\phi}}_{\widetilde{\mathcal{S}}_\phi} - 
                    \bm{X}_{\mathcal{S}_\beta} \bm{\beta}_{\mathcal{S}_\beta} - 
                    \bm{I}_{n, \mathcal{S}_\phi} 
                    \bm{\phi}_{\mathcal{S}_\phi} \|^2_2 }{ 
                    n \max(\lvert \mathcal{S}_\phi \backslash \widetilde{\mathcal{S}}_\phi \rvert + 
                    \lvert \mathcal{S}_\beta \backslash \widetilde{\mathcal{S}}_\beta \rvert, 1 ) } 
                $$
                where $\widehat{\bm{\phi}}_{\widetilde{\mathcal{S}}_\phi}$ is any estimate such that 
                $\widetilde{\mathcal{S}}_\phi \neq \mathcal{S}_\phi  ,  \lvert \widetilde{\mathcal{S}}_{\phi} \rvert \leq m_M $
                and 
                $ \widehat{\bm{\beta}}_{\widetilde{\mathcal{S}}_\beta} $ 
                satisfies
                $\lvert {\widetilde{\mathcal{S}}_\beta} \rvert \leq p_0$.
        \end{condition}
        \noindent 
        Conditions~\ref{cond2}(A)-(C) are quite common for nonconcave penalization methods such as SCAD \citep{fan2012variable}, and they are based only on the set of non-outlying cases indexed by $\mathcal{S}_\phi^c$. 
        Condition~\ref{cond2}(D) is specifically required to detect MSOM outliers based on $L_0$-constraints \citep{insolia2020simultaneous}. It bounds the difficulty of MSOMs detection based on a \textit{minimal degree of separation} between the true and a least favorable model. 
        Intuitively, it requires that MSOM outliers have larger residuals for models of comparable sizes. 
        This relates to the signal-to-noise-ratio and it improves the heuristic argument $n>5p$ which is often advocated for robust estimation methods \citep{rousseeuw1990unmasking}.
        The following result ensures that our proposal provides simultaneous feature selection and MSOM outlier detection consistency.
        \begin{theorem}[Robust weak oracle property] \label{thm:fix}
            Under all conditions in lists~\ref{cond1} and \ref{cond2}, and that $\log p = o( (n - m_M) \lambda^2 )$ and  $\sqrt{n-m_M} \lambda  \to \infty $ as $ (n - m_M) \to \infty$. 
            Then, there exist $k_n$ and a strict local minimizer of  \eqref{eq:reg_fixVIOM_MSOM} such that the resulting robust estimates achieve:
            \begin{enumerate}
                \item \textit{Sparsity}: 
                    $P \left( \widehat{\bm{\beta}}_{\widehat{\mathcal{S}}_\beta^c} = \bm{0} \right) \to 1 $;
                \item \textit{Bounded $L_\infty$-norm}: 
                    $P \left ( \| \widehat{\bm{\beta}}_{\widehat{\mathcal{S}}_\beta} - \bm{\beta}_{\mathcal{S}_\beta} \|_\infty < (n - m_M)^\tau \log (n - m_M) \right) \to 1 $;
                \item \textit{MSOM detection}: 
                    $P \left( \widehat{\mathcal{S}}_\phi = \mathcal{S}_\phi \right) \geq P \left( \widehat{\bm{\phi}} = \bm{\phi}_0 \right) \to 1$.
            \end{enumerate} 
        \end{theorem}
        \noindent
        Here the number of features in $\bm{\beta}$ is allowed to exponentially increase with the (uncontaminated) sample size $n - m_M$. 
        This is a robust version of the weak oracle property in the sense of \citet{lv2009unified} and \citet{fan2012variable}. 
        
        We remark that existing robust model selection procedures, which explicitly consider only MSOM outliers, can be cast into \eqref{eq:equivalLMMmsomViom2}. However, differently from \eqref{eq:reg_fixVIOM_MSOM}, they do not take into account the random structure of the problem, such as VIOM outliers.
        Relatedly, our approach  can be naturally extended to high-dimensional mixed-effects linear models; 
        however, this is left for future work.
        Moreover, regardless the presence of VIOMs, the use of nonconcave penalties in \eqref{eq:reg_fixVIOM_MSOM} provides an important bridge between existing trimming estimators, which promote sparsity in the feature space based on convex penalties \citep{kurnaz2017robust,alfons2013sparse}, and the optimal approach based on $L_0$-constraints \citep{insolia2020simultaneous}. 
        Unlike the former, our proposal achieves oracle properties under weaker assumptions, which can be particularly useful for the latter;
        e.g.,~to provide better warm-starts and big-$\mathcal{M}$ bounds, and accelerate convergence for MIP techniques.

    \subsection{Step 2: VIOM Detection}
    \label{secsub:proposal_random}
    	    
        VIOM outlier detection, based on sparse estimation of $\bm{\gamma}$ in \eqref{eq:equivalLMMmsomViom2}, differs from sparse estimation of fixed effects ($\bm{\beta}$ and $\bm{\phi}$) due to their intrinsic randomness.
        Indeed, while underfitting $\bm{\gamma}$, which results in undetected VIOMs, introduces bias in the estimated variance for the fixed effects in $\bm{\beta}$, the inclusion of irrelevant $\bm{\gamma}$ components, i.e.,~wrongly detected VIOMs, decreases the estimator efficiency. 
        
        In this section, based on the results from Section~\ref{secsub:proposal_fixed}, we consider the augmented design matrix
        $
        \overline{\bm{X}} = [ 
        \bm{X}_{\widehat{\mathcal{S}}_\beta} , 
        \bm{D}_{\widehat{\mathcal{S}}_\phi}
        ]
        $,
        where $ \bm{X}_{\widehat{\mathcal{S}}_\beta} $ and $ \bm{D}_{\widehat{\mathcal{S}}_\phi} $ index the estimated $k_p$ active features and $k_n$ MSOM outliers, respectively.
        We further assume that $ n - k_n \geq k_p $, and that
        $
        \overline{\bm{X}}^T \overline{\bm{X}}
        $
        is an invertible matrix of size $ (k_p +k_n) $. 
        The corresponding matrix of error contrasts is denoted as $\overline{\bm{A}}$, and $ \bm{P}_{\overline{x}} $ is the counterpart of $ \bm{P}_x $ using $ \overline{\bm{X}} $ in place of $ \bm{X} $. 
            
        Based on REMLE theory, the conditional distribution $f(  \overline{\bm{A}}^T \bm{y} \lvert  \bm{\gamma}_{\mathcal{S}_\gamma} )$ does not depend on $\bm{\beta}$, $\bm{\phi}$ and $\overline{\bm{A}}$, which leads to the restricted posterior density
        \begin{align} \label{eq:rpd}
            f \left(  \bm{\gamma}_{\mathcal{S}_\gamma} \lvert  \overline{\bm{A}}^T \bm{y} \right) &= 
            f  \left(  \overline{\bm{A}}^T \bm{y} \lvert \bm{\gamma}_{\mathcal{S}_\gamma} \right) f(\bm{\gamma}_{\mathcal{S}_\gamma} )  \nonumber \\ 
            &= (\bm{y}  - 
            \bm{D}_{\mathcal{S}_\gamma}
            \bm{\gamma}_{\mathcal{S}_\gamma} )^T \bm{P}_{\overline{x}} 
            (\bm{y}-  \bm{D}_{\mathcal{S}_\gamma} \bm{\gamma}_{\mathcal{S}_\gamma}) + \bm{\gamma}^T_{\mathcal{S}_\gamma} 
            \bm{\Gamma}^{-1} \bm{\gamma}_{\mathcal{S}_\gamma} .
        \end{align}
        However, \eqref{eq:rpd} cannot be used to estimate $ \bm{\gamma}$ as it relies on the unknown set of VIOM outliers $ \mathcal{S}_\gamma $, as well as their covariance matrix $\bm{\Gamma}$.
        We replace \eqref{eq:rpd} with the following objective function
        \begin{align} \label{eq:reg_randVIOM_MSOM} 
            \widehat{\bm{\gamma}} = \operatorname*{\arg\min}_{\bm{\gamma}} ~&(\bm{y}- \bm{\gamma})^T \bm{P}_{\overline{x}} (\bm{y}- \bm{\gamma})  + \bm{\gamma}^T \bm{\mathcal{M}}_\gamma^{-1} \bm{\gamma}  
            + 
            (n - k_n) \sum_{i \in \widehat{\mathcal{S}}_\phi^c} R_{\lambda}(\lvert \gamma_i \lvert)  
        \end{align}
        where $\bm{\mathcal{M}}_\gamma $ is a proxy for $\bm{\Gamma}$ (see the Supplementary Material for details).
        In principle the penalty function $R_\lambda(\cdot)$  might differ from the one in \eqref{eq:reg_fixVIOM_MSOM}, but for simplicity we consider nonconcave penalties such as SCAD also here.
            
        In order to control the bias for the oracle-assisted estimator $ \gamma_i^2 / (n-m_M) $ of $ \sigma^2 \omega_i$, we condition on the event $ \{  \min_{i \in \mathcal{S}_\gamma } \lvert \gamma_i \rvert \geq \sqrt{n-m_M} b_0^* \} $, where $b_0^* \in (0, \min_{i \in \mathcal{S}_\gamma } \sigma \sqrt{ \omega_i } ) $ and
        $ \omega_i = \text{var}(\gamma_i) / \sigma^2 $.
        Let $\mathbf{P}_{\overline{x}}^{\mathcal{S}_\gamma}$ comprise the rows and columns of $\mathbf{P}_{\overline{x}}$ belonging to the VIOM outliers in $\mathcal{S}_\gamma$.
        We rely on the following conditions to detect such outliers.
        \begin{condition}[VIOM reconstruction] \label{cond3}
        \item[A.]
        \underline{Design matrix and VIOM outliers}:
        for some constant $c_{3} > 0$, the minimum and maximum eigenvalues of $ (n - m_{M})^{-1}  \mathbf{P}_{\overline{x}}^{\mathcal{S}_\gamma} $ and $\bm{\Gamma}$ are bounded from below and above, respectively, by $c_{3}$ and $c_{3}^{-1}$.
        Moreover, there exists $\delta \in (0, 1/2 )$ such that
        $$
        \left\| (\mathbf{P}_{\overline{x}}^{\mathcal{S}_\gamma} + \bm{\Gamma}^{-1} )^{-1} \right\|_{\infty} \leq \frac{ ( n - m_M)^{-(1+\delta)/2}}{R_{\lambda}^{\prime}\left(\sqrt{n - m_M} b_0^* / 2\right)} ,
        $$
        $$
        \max _{i \in \mathcal{S}_{\gamma}^{c} \cap \mathcal{S}_{\phi}^{c} }\left\|
        \mathbf{P}_{\overline{x}, i}  
        \bm{D}_{\mathcal{S}_\gamma}
        (\mathbf{P}_{\overline{x}}^{\mathcal{S}_\gamma} + \bm{\Gamma}^{-1} )^{-1}
        \right\|_{2}<\frac{R_{\lambda}^{\prime}(0+)}{R_{\lambda}^{\prime}\left(\sqrt{n - m_M} b_0^* / 2\right)} .
        $$
                    
        \item[B.] 
        \underline{VIOM strength}: 
        $\sup _{\left\{t \geq \sqrt{n - m_M} b_0^* / 2\right\}} R_{\lambda}^{\prime \prime}(t)=o\left((n - m_M)^{-1}\right)$.
                    
        \item[C.]
        \underline{Proxy matrix}:
        $\Lambda_{\min }\left(\bm{\mathcal{M}}_\gamma^{\mathcal{S}_\gamma^c} \right) \geq 0$
        and 
        $\Lambda_{\min }\left(\bm{\mathcal{M}}_\gamma^{\mathcal{S}_\gamma} - \bm{\Gamma} \right) \geq 0$.
        
        \end{condition}
        \noindent
        Similar conditions can be found in \citet{fan2012variable} to perform feature selection on random effects using nonconcave penalties.
        The following result shows that our proposal detects VIOM outliers with asymptotic probability one, and effectively down-weights them.
        \begin{theorem}[VIOM treatment] \label{thm:rnd}
            Under all  conditions in lists \ref{cond1}-\ref{cond3}, and  that $ b_0^* (n - m_M)^{\delta - 1/2} \to \infty$ as $(n - m_M) \to \infty$, there exists $\lambda$ such that a strict local minimizer of \eqref{eq:reg_randVIOM_MSOM}  satisfies:
            \begin{enumerate}
            \item 
            \textit{VIOM detection}: 
            $ P \left( \widehat{\mathcal{S}}_\gamma = \mathcal{S}_\gamma \right) \to 1 $;
                        
            \item 
            \textit{VIOM down-weighting}: 
            $ \max_{i \in \mathcal{S}_{\gamma}}
                        \| \widehat{\gamma}_{i} -
                        \gamma_{i}
                        \| 
                        \leq (n - m_M)^{-\delta} 
                        $
                        for $\delta \in (0, \frac{1}{2} )$.
            \end{enumerate}
        \end{theorem}

    \subsection{Step 3: Weights Estimation}
    \label{subsec:step3}
            
Steps 1 and 2 described above might induce non-negligible biases, especially in a finite-sample setting. 
To mitigate such biases, we propose an {\it ex-post} update for the VIOM outlier weights and other regression parameters depending on them. This is similar in spirit to post-selection updates implemented with feature selection methods; e.g.,~lasso followed by an OLS fit restricted to the set of active features \citep{liu2013asymptotic}.
            
Specifically, we consider a feasible counterpart of the mixed-effects linear model in \eqref{eq:equivalLMMmsomViom}, which is based on the estimated sets $ \widehat{ \mathcal{ S } }_\phi $ and $\widehat{\mathcal{S}}_\gamma $ (MSOM and VIOM outliers), and $\widehat{\mathcal{S}}_\beta $ (active features). 
We first remove the units belonging to $ \widehat{ \mathcal{ S } }_\phi $ from the fit, and apply REMLE to estimate weights for the units in $\widehat{\mathcal{S}}_\gamma $ conditionally on the features in $\widehat{\mathcal{S}}_\beta $.
Next, we use these weights to update the estimates of $\bm{\beta}_{\widehat{\mathcal{S}}_\beta}$. 
This approach guarantees that, if Steps 1 and 2 identify the true model in terms of features ($\mathcal{S}_\beta$) as well as outliers ($\mathcal{S}_\phi$ and $\mathcal{S}_\gamma$), then our proposal reaches an optimal trade-off between breakdown point and efficiency.
            
The following definition extends the robustly strong oracle property in the sense of \citet{insolia2020simultaneous} to the concurrent presence of MSOM and VIOM outliers.
\begin{definition}[Doubly robust strong oracle property]
Let $\mathcal{S} = \{ \mathcal{S}_\beta , \mathcal{S}_\phi, \mathcal{S}_\gamma  \} $,
and define the doubly robust strong oracle estimator
$         
    \widehat{\bm{\beta}}_{S} = 
    \widehat{\bm{\beta}} \lvert \mathcal{ S }
$
as the solution for $ \bm{\beta} $ in \eqref{eq:equivalLMMmsomViom}.
An estimator $\widehat{\bm{\beta}}_{\widehat{\mathcal{S}}} $ satisfies the doubly robust strong oracle property if there exist tuning parameters which ensure 
$ 
P ( \widehat{\mathcal{S}}= \mathcal{S} )
    \geq 
    P (  \widehat{\bm{\beta}}_{\widehat{\mathcal{S}}} = \widehat{\bm{\beta}}_{S} )
    \to 1  
$
in the presence of MSOM and VIOM outliers.
\end{definition}

\noindent
The following result refines Theorems~\ref{thm:fix} and~\ref{thm:rnd}, and ensures that our proposal achieves the doubly robust strong oracle property -- allowing us to rely on large sample inference.
    \begin{theorem}[Doubly robust strong oracle property] \label{thm:opt_bdp_eff}
    Under all conditions in lists~\ref{cond1}-\ref{cond3}, 
    as $(n - m_M) \to \infty$ there exist tuning parameters $k_n$ and $\lambda$'s in \eqref{eq:reg_fixVIOM_MSOM} and \eqref{eq:reg_randVIOM_MSOM} such that the resulting estimator plugging $\widehat{\mathcal{S}}$ in \eqref{eq:equivalLMMmsomViom} achieves:
        \begin{enumerate}
        \item 
        \textit{Asymptotic unbiasedness}: 
                $$
                    \| E \widehat{\bm{\beta}} - \bm{\beta}_0  \|_2^2 \leq 
                     2 P( \widehat{\mathcal{S}} \neq \mathcal{S} )  
                     \left\{ 
                     \|  \bm{\beta}_0
                    \|_2^2 
                    + 
                    \lambda_M
                    \left( \| \widehat{\bm{W}}^{1/2} \bm{X} \bm{\beta}_0 \|_2^2 + \sigma^2 \tr(\widehat{\bm{W}} )  \right)
                     \right\}
                     \to 0
                $$
        where 
        $ \tr(\cdot)$ is the matrix trace,
        $
        \lambda_M =   \Lambda_{\max} \{ (\bm{X}_{\widetilde{\mathcal{S}}_\beta}^T \widehat{\bm{W}} \bm{X}_{\widetilde{\mathcal{S}}_\beta} )^+ \} > 0
        $
        and 
        $\{ \widetilde{\mathcal{S}}_\beta : \widehat{\mathcal{S}}_\beta \neq \mathcal{S}_\beta \} $.
                    
        \item 
        \textit{Optimal MSE}: 
            \begin{align}
                     E \|  \widehat{\bm{\beta}} - \bm{\beta}_0  \|_2^2 &\leq
                      \sigma^2 \tr(\bm{\Sigma}_{X}^{-1}) / \tr(\widehat{\bm{W}})
                        \nonumber \\
                       &+
                     2 P( \widehat{\mathcal{S}} \neq \mathcal{S} )  
                     \left\{ 
                    (\lambda_M + \lambda_{M_s} )
                    \left( \| \widehat{\bm{W}}^{1/2} \bm{X} \bm{\beta}_0 \|_2^2 + \sigma^2 \tr(\widehat{\bm{W}} )  \right)
                     \right\} \nonumber
            \end{align}
        where
        $ 
        \lambda_{M_s} = \Lambda_{\max}\{ (\bm{X}_{\mathcal{S}_\beta}^T \widehat{\bm{W}} \bm{X}_{\mathcal{S}_\beta} )^{-1} \} 
        $
        and
        $
        \bm{\Sigma}_{X} = 
                (\mathbf{X}_{\mathcal{S}_\beta}^T
                    \widehat{\bm{W}}
                    \mathbf{X}_{\mathcal{S}_\beta} ) .
        $
                    
        \item 
        \textit{Asymptotic normality}: 
                $ 
                       \sqrt{n} ( \widehat{\bm{\beta}} - \bm{\beta}_0 ) \to^d N ( \bm{0},  \sigma^2 ( \bm{\Sigma}_{X} / n)^{-1} ) .  
                $
        \end{enumerate} 

    \end{theorem}
            
\noindent
Importantly, this result provides also some intuition on the estimator's behavior when it does not retrieve the doubly robust oracle solution, as well as in finite-sample settings.
Indeed, points 1 and 2 in Theorem~\ref{thm:opt_bdp_eff} depend on the probability of not recovering the true model, in terms of active features and/or outlying cases -- which increases estimation biases and MSE.
Finally, weights estimates obtained in Step~3 can be used to update the proxy matrices used in Sections~\ref{secsub:proposal_fixed} and \ref{secsub:proposal_random}, suggesting an iterative strategy whereby the process in Steps 1-3 is repeated improving model selection and estimation results (see Section~\ref{sec:sim}). 
A similar approach was proposed in \citet{fan2012variable} to select and estimate fixed and random effects; here our iteration includes an additional third step to update the weights.

\subsection{A Heuristic Procedure} 
\label{secsub:proposal_heuristic} 

Here we present a computationally lean heuristic procedure similar to two-stage regression for mixed-models, which is inspired by our main proposal; 
namely:
    \begin{enumerate}
    \item
    Solve \eqref{eq:reg_fixVIOM_MSOM} using the proxy matrix $\bm{\mathcal{M}}_R = \bm{I}_n$.
    Let $\bm{y}^* = \bm{y}_{\widehat{\mathcal{S}}_\phi^c}$ and $\bm{X}^* = \bm{X}_{\widehat{\mathcal{S}}_\phi^c, \widehat{\mathcal{S}}_\beta }$ comprise response and predictor values restricted to the selected relevant features and non-outlying cases.

    \item
    Consider again \eqref{eq:reg_fixVIOM_MSOM} using $\bm{y}^*$, $\bm{X}^*$ and $\bm{\gamma}_{\widehat{\mathcal{S}}_\phi^c}$ in place of $\bm{y}$, $\bm{X}$ and $ \bm{\phi} $, respectively.
    Using $\bm{\mathcal{M}}_R = \bm{I}_{n-k_n}$ and leaving the estimation of $\bm{\beta}$ unpenalized, solve the model relaxing the $L_0$-constraint (e.g.,~using SCAD or lasso). Let $ \widehat{\bm{\gamma}}_{\widehat{\mathcal{S}}_\gamma} $ indicate the resulting sparse estimates.
                    
    \item 
    Consider $\bm{y}^*=\bm{X}^* \bm{\beta} + \bm{\epsilon} $ and, similar to Section~\ref{subsec:step3}, estimate weights for the units $ i \in {\widehat{\mathcal{S}}_\gamma} $ using REMLE and use WLS to update the estimation of $\bm{\beta}$.
    \end{enumerate}
Step 1 can be efficiently tackled using sparse high-breakdown point estimators based on heuristics. It detects MSOMs (i.e.,~it estimates non-zero entries in $\bm{\phi}$) and selects active features in $\bm{\beta}$.
Step 2, which is related to ridge regression (see the Supplementary Material for details), is used to detect VIOMs. This is equivalent to assuming a MSOM if the active $\bm{\gamma}$ coefficients are not shrunk (e.g.,~using $L_0$-constraints these units receive zero weights). Otherwise units are down-weighted or left with their full weights; we follow this approach as MSOMs are detected in Step 1.
Step 3, which might be skipped if one is only interested in $\bm{\beta}$, is useful to reduce possible biases introduced in Steps 1-2, and in principle might be combined with Step 2 (see again the Supplementary Material for details).
                
We remark that Steps~1 and~2 of our heuristic procedure require a careful tuning process, which is critical to estimate the weights in a data-driven fashion and guarantee their ``adaptiveness'' (i.e.,~the breakdown point and the efficiency of the corresponding $\bm{\beta}$ estimates). 
In the Supplementary Material we describe the robust BIC proposed for this tuning, and discuss connections between our heuristic procedure, ridge and $M$-estimation.

\section{Simulation Study} 
\label{sec:sim}
        
In this section we compare our proposal with state-of-the-art methods through numerical simulations. 
The data is generated as follows. Each row of the $n \times p$ design matrix $\bm{X}$ contains a $1$ (for the intercept), and then entries drawn independently from a $N(\bm{0}, \bm{I}_{p-1})$.
The $p$-dimensional coefficient vector $\bm{\beta}$ contains $p_0$ non-zero entries (including the intercept), and the errors $\varepsilon_i$ are drawn independently from a $N (0, \sigma^2_{\text{SNR}})$. 
$\sigma^2_{\text{SNR}}$ depends on the signal-to-noise-ratio $  \text{SNR}  = \text{var}(\bm{X}\bm{\beta}) / \sigma^2_{\text{SNR}}$ and controls the difficulty of the problem.
Then, $m_V$ and $m_M$ points out of $n$ are contaminated as in \eqref{eq:cont_model}. 
Mean shifts affect error and active predictors in the design matrix, with strengths $\mu_\varepsilon $ and  $\mu_X $, respectively.
Variance inflation affects only the error, with a common parameter $v$.
Each simulation scenario is replicated $t$ times and results are averaged. 
    	
We consider the following performance metrics:
{\bf (i)} MSE of $ \widehat{\bm{\beta}} $ partitioned into variance and squared bias. 
For each estimated coefficient 
        \begin{equation} \label{eq:metricMSE}
        \text{MSE}(\widehat{\beta}_j) = \frac{1}{t} \sum_{i=1}^{t} (\widehat{\beta}_{ij} - \beta_j)^2 =
        		\frac{1}{t} \sum_{i=1}^{t} (\widehat{\beta}_{ij} - \overline{\beta}_j)^2 +
        		(\overline{\beta}_j - \beta_j)^2 , 
        \end{equation}
where $ \overline{\beta}_j = \frac{1}{t} \sum_{i=1}^{t} \widehat{\beta}_{ij}$, and we average the MSE across coefficients to produce $ \text{MSE}(\widehat{\bm{\beta}}) = \frac{1}{p} \sum_{j=1}^{p} \text{MSE}(\widehat{\beta}_j)$.
{\bf (ii)} For low-dimensional settings without MSOMs, we also consider the MSE of a weighted estimate of the error variance 
        \begin{equation} \label{eq:metrics2}
        \widehat{s}^2 = \frac{1}{(n-p)} \frac{\sum_{i=1}^{n} \widehat{w}_i 
        e_i^2}{ \sum_{i=1}^{n} \widehat{w}_i / n  } , 
        \nonumber
        \end{equation}
where the $ e_i$'s are the raw residuals and the $ \widehat{w}_i $'s the estimated weights. This takes into account weight estimates regardless of whether some units are in fact contaminated. 
The MSE decomposition for $ \widehat{s}^2 $ is computed as in \eqref{eq:metricMSE}, with $ \sigma^2_{\text{SNR}} $ and $ \widehat{s}^2 $ replacing $ \bm{\beta} $ and $ \widehat{\bm{\beta}} $, respectively.
{\bf (iii)} Let the non-zero entries of $\bm{\tau} = \bm{\phi} + \bm{\gamma}$ indicate MSOMs and/or VIOMs. Outlier detection accuracy is measured in terms of false positive and false negative rates
        \begin{align}
        &\text{FPR}(\widehat{\bm{\tau}})=\frac{\left|\left\{i \in\{1, \ldots, n\}: \widehat{\tau}_{i} \neq 0 \wedge \tau_{i}=0\right\}\right|}{\left|\left\{i \in\{1, \ldots, n\}: \tau_{i}=0\right\}\right|}  ,  \label{eq:metricFPR} \\
        &\text{FNR}(\widehat{\bm{\tau}})=\frac{\left|\left\{i \in\{1, \ldots, n\}: \widehat{\tau}_{i}=0 \wedge \tau_{i} \neq 0\right\}\right|}{\left|\left\{i \in\{1, \ldots, n\}: \tau_{i} \neq 0\right\}\right|} .   \label{eq:metricFNR} 
        \end{align}
These indicate the proportion of uncontaminated units wrongly detected as outliers, and of undetected contaminated units, respectively.
{\bf (iv)} For sparse settings, we also consider feature selection accuracy -- which is measured in terms of FPR and FNR as in \eqref{eq:metricFPR} and \eqref{eq:metricFNR}, using $\beta_j$ and $\widehat{\beta}_j$ (for $j=1,\ldots,p$) in place of $\tau_i$ and $\widehat{\tau}_i$, respectively.

\subsection{Scenario 1: Low-Dimensional VIOMs}
    
Here we set $p=p_0=2 $, with $\bm{\beta} = (2, 2)^T$ and $\text{SNR} = 3$. 
The proportion of VIOM outliers is $m_V/n=0.25$ and $ v=10$. The sample size $n$ increases from 50 to 500 with 10 equispaced values. Data for each setting are replicated $ t = 100 $ times.
    	
We consider the \textit{oracle benchmark} (Opt), i.e.,~a WLS fit based on the true population weights $\bm{w}$, along with: 
{\bf (a)} OLS, the ordinary least squares estimator
{\bf (b)} LTS, the least trimmed sum of squares estimator with trimming set to the true $ m_V/n$ \citep{maronna2006robust};
{\bf (c)} MM85, an MM-estimator using a preliminary LTS and Tukey's bisquare loss function, with tuning constant set to achieve 85\% nominal efficiency \citep{maronna2006robust};
{\bf (d)} MM95, as in (c), with 95\% nominal efficiency;
{\bf (e)} FSRws, which utilizes a variant of forward search and single REMLE weights as described in \citet{insolia2020ViomMsom};
{\bf (f)} Heur, our heuristic procedure (Section~\ref{secsub:proposal_heuristic}), where in Step~2 $\bm{\gamma}$ is estimated by adaptive lasso initialized with OLS residuals, and in Step~3 each weight is estimated independently using REMLE as in FSRws; 
{\bf (g)} SCADws, our main proposal (Section~\ref{sec:proposal}), where in Step~3 weights are estimated by a REMLE fit on the active random components of $ \bm{\gamma}$ detected by SCAD -- as in FSRws and Heur, these weights are estimated independently. 

Figure~\ref{fig:simRes1bet} shows the MSE for $\widehat{\bm{\beta}}$; SCADws and MM85 generally outperform other methods, Heur and MM95 perform comparably,  FSRws improves on LTS and OLS (which perform poorly across sample sizes).
Figure~\ref{fig:simRes1sig} shows the MSE for $\widehat{s}^2$.
Notably, SCADws generally outperforms other methods, including the oracle estimator -- likely because some VIOM outliers which are down-weighted by the latter do not carry sizeable residuals. Nevertheless, SCADws is capable of estimating full wights for these points. Relatedly, non-outlying cases with large residuals by chance are given full weight by the oracle estimator, but not necessarily by SCADws (see circled dots on the right panel of Figure~\ref{fig:simtimeFit}).
MM85 outperforms MM95, highlighting the drawbacks of $M$-estimators with pre-specified efficiency values.
Heur performs comparably, although its estimates have larger biases, and it outperforms LTS and OLS, which provide strongly biased estimates because each point receives a binary or full weight.
The performance of FSRws decreases for smaller sample sizes, where outliers are more often undetected.
\begin{figure}[ht!]
    		\centering
			\includegraphics[width=0.85\linewidth]{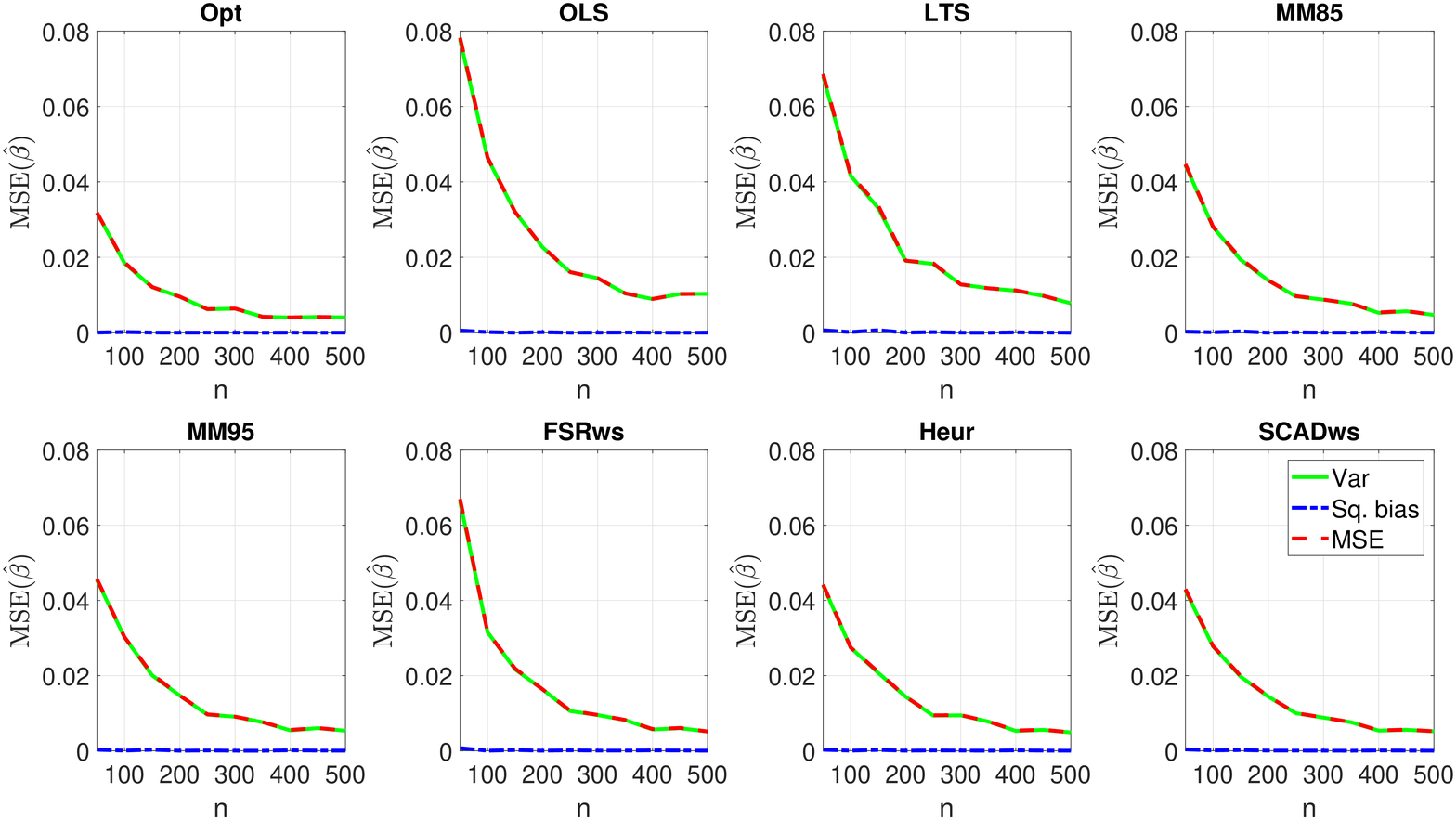}
    		\caption{Scenario 1. $\text{MSE}(\widehat{\bm{\beta}})$ comparisons across procedures and sample sizes.}
    		\label{fig:simRes1bet}
    	\end{figure}
	 	\begin{figure}[ht!]
    		\centering
			\includegraphics[width=0.85\linewidth]{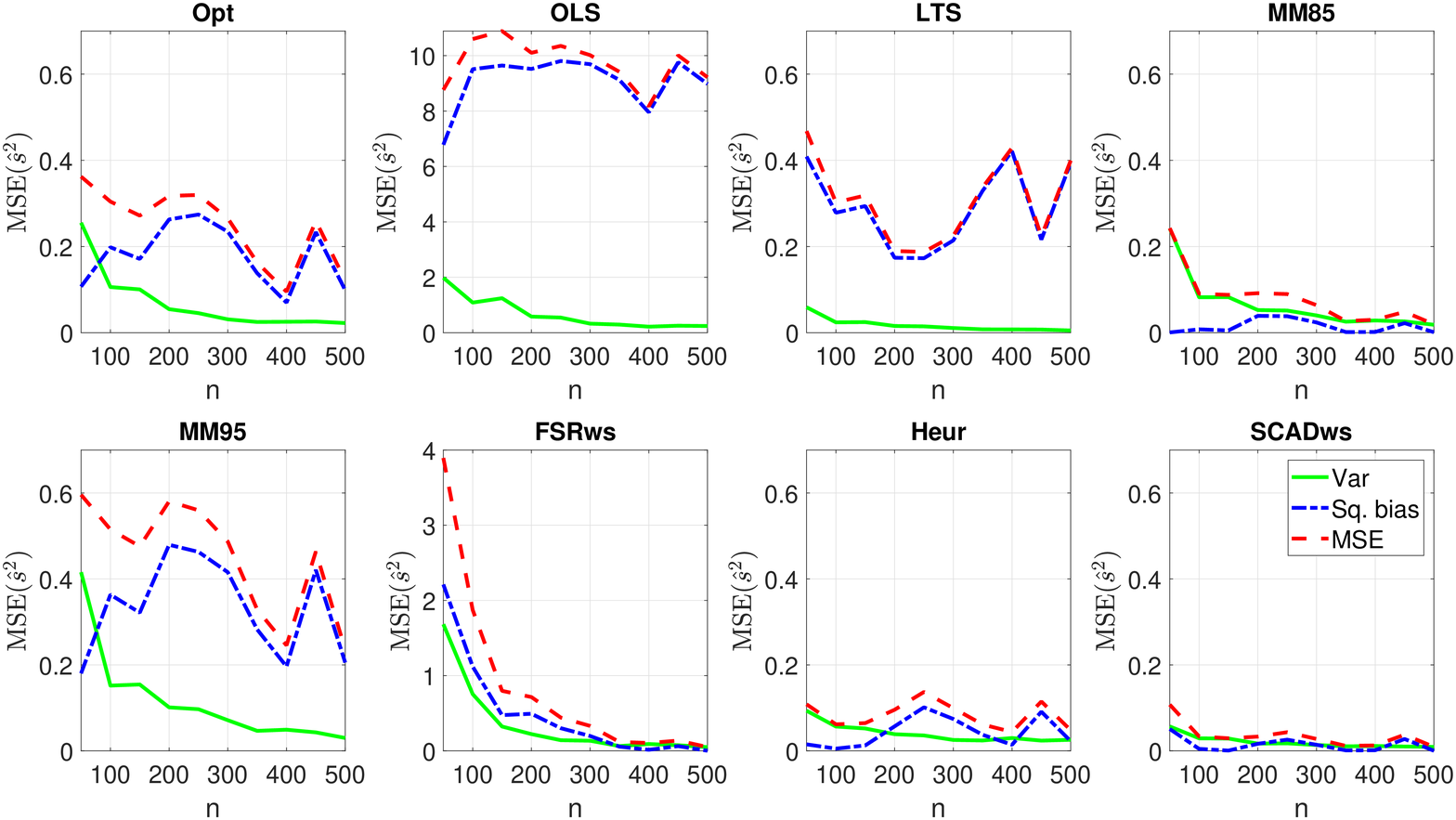}
    		\caption{Scenario 1. $\text{MSE}(\widehat{s}^2)$ comparisons across procedures and sample sizes.}
    		\label{fig:simRes1sig}
    	\end{figure}
        
The two left panels of Figure~\ref{fig:simtimeFit} show FPR and FNR for VIOM detection across methods, respectively.
Overall, SCADws outperforms other methods; its decrease in terms of FPR along sample sizes is partially compensated by an increase in FNR. FSRws is close to SCADws for larger sample sizes, but for smaller ones it fails to detect some outliers (low FPR and high FNR).
Heur performs similarly to SCADws, and MM-estimators perform poorly in these metrics due to a general down-weighting of all units.
These trends demonstrate the ability of SCADws to detect truly outlying cases as the sample size increases. On the other hand, while FSRws tends to be more conservative across sample sizes, LTS has a more aggressive behavior resulting in larger FPR and lower FNR.
The right panel of Figure~\ref{fig:simtimeFit} shows a scatterplot summarizing results for a typical simulation ($n=500$). 
True VIOM outliers, as well as the ones detected by SCADws, are highlighted. 
         \begin{figure}[ht!]
    		\centering
    		\begin{subfigure}{.55\textwidth}
    			\centering
    			\includegraphics[width=1\linewidth]{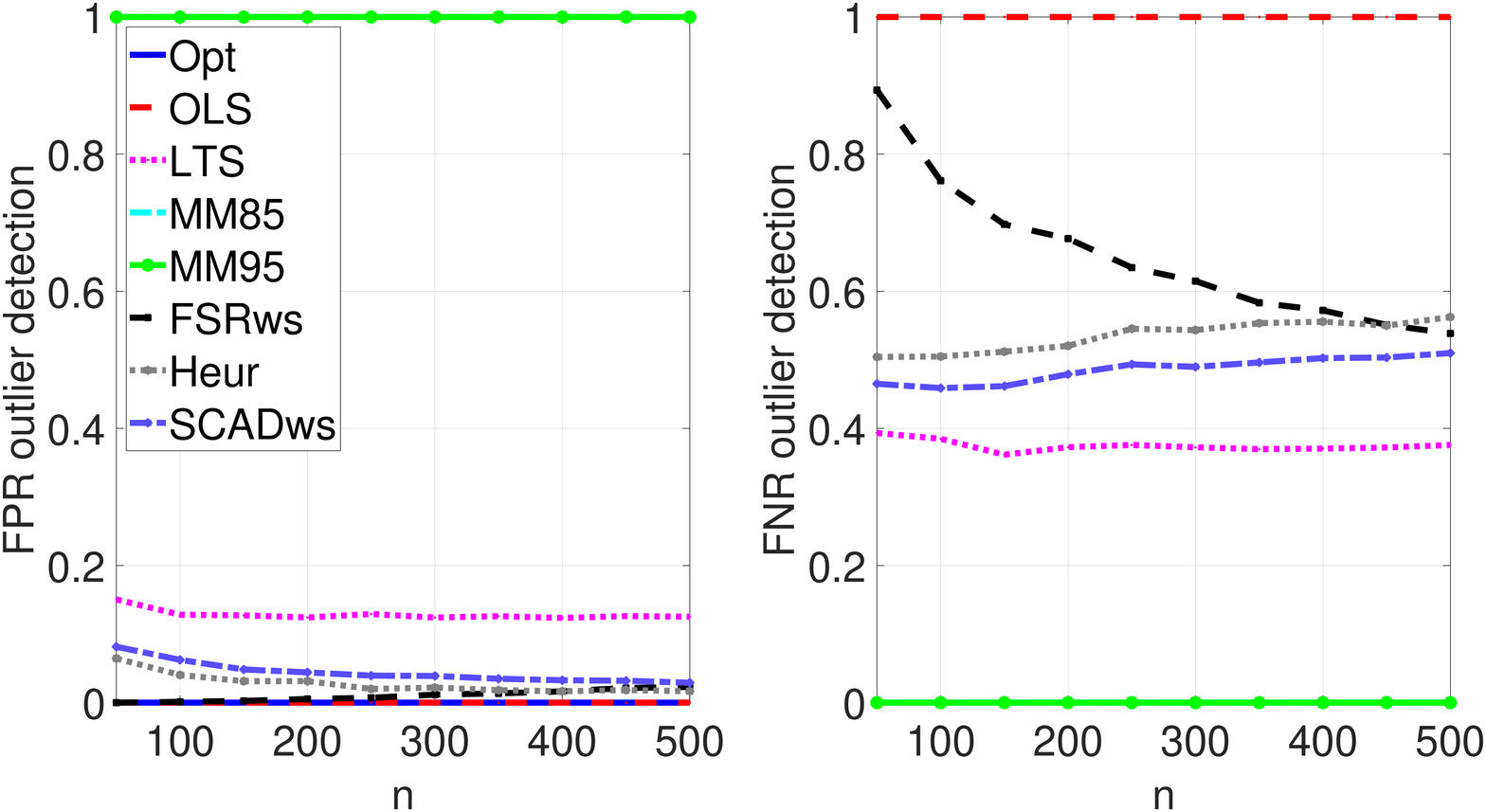}
    		\end{subfigure}%
    		\begin{subfigure}{.45\textwidth}
    			\centering
			\includegraphics[width=0.8\linewidth]{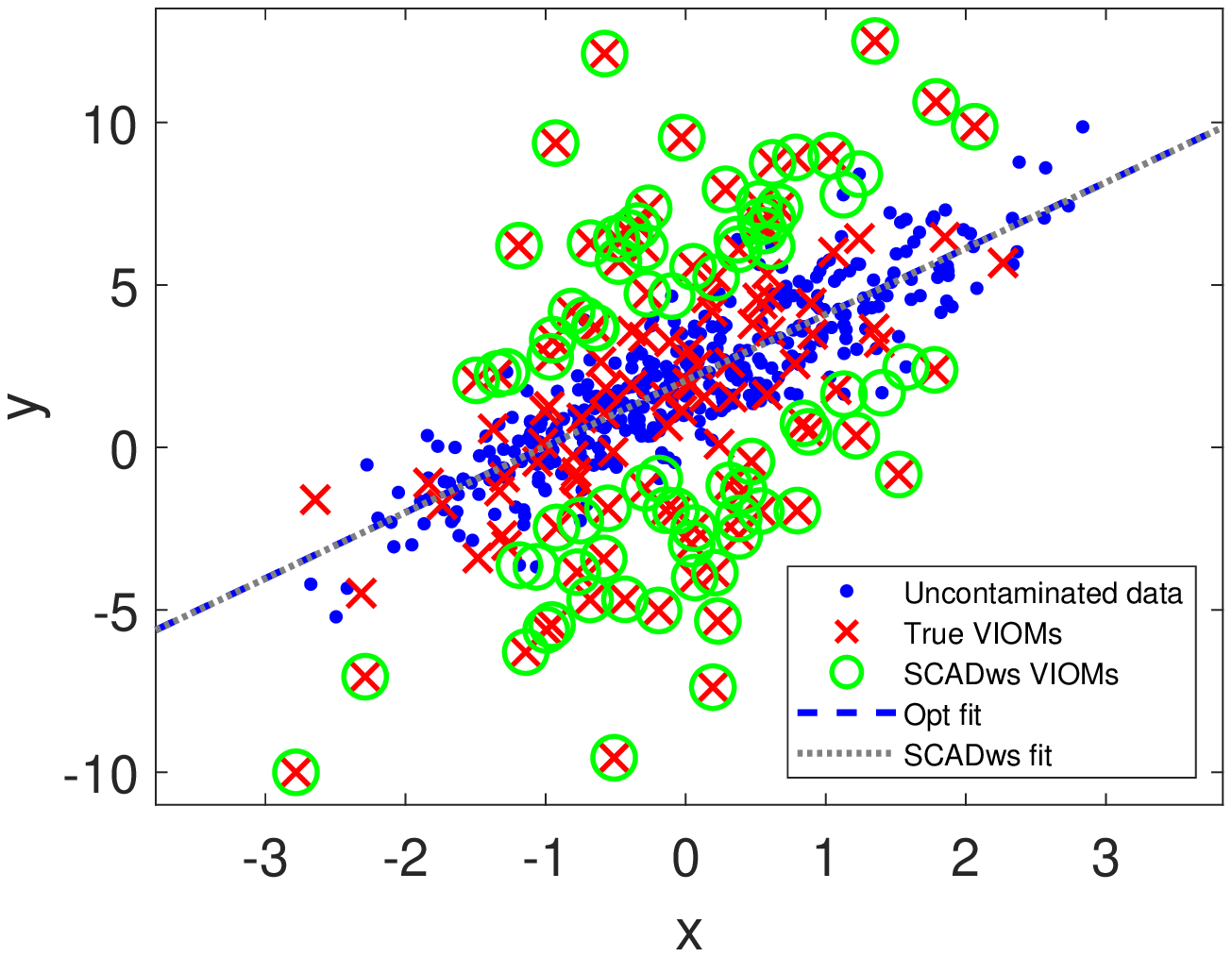} 
    		\end{subfigure}
    		\caption{Scenario 1. Left: comparisons of FPR and FNR for outlier detection across procedures and sample sizes.
    		Right: scatterplot summarizing results for a typical simulation with $n=500$ -- true VIOMs and VIOMs detected by SCADws are highlighted. 
    		}
    		\label{fig:simtimeFit}
    	\end{figure}

\subsection{Scenario 2: High-Dimensional VIOMs and MSOMs}
        
Here we mimic Scenario~1, but we use sparse fixed effects in $\bm{\beta}$ and introduce MSOM outliers. 
Specifically, we set $p = 30$ with $p_0 = 3$ active features. 
The proportions of VIOM and MSOM outliers are set to $m_V/n=0.15$ and $m_M/n=0.05$. Mean shifts are set to $\mu_\varepsilon = -10$ and $\mu_X = 10$ in order to create bad leverage points. The sample size $n$ ranges from 60 to 150 (with 10 equispaced values). 
Data for each setting are again replicated $ t = 100 $ times.

The oracle benchmark (Opt) is computed using population weights and the active feature set.
In addition to it, we consider:
{\bf (a)} lasso; 
{\bf (b)} sparseLTS \citep{alfons2013sparse}; 
{\bf (c)} TaL, adaptive lasso with Tukey's bisquare loss, a preliminary sparseLTS fit, and tuning constant fixed to achieve 85\% nominal efficiency \citep{chang2018robust};
{\bf (d)} Heur, as in Scenario~1, but with a preliminary fixed-effects selection and MSOM detection using robust SCAD. 
{\bf (f)} SCADws, as in Scenario~1, but with a preliminary fixed-effects selection and MSOM detection based on \eqref{eq:reg_fixVIOM_MSOM};
{\bf (g)} SCAD2s, two iterations of SCADws where weights estimated in the first iteration are used to update the proxy matrices and re-run our 3-step procedure;
{\bf (h)} SCADopt, similar to SCADws, but with proxy matrices built with VIOM population weights;
For simplicity, robust methods all use the true trimming level $m_M/n$.

Figure~\ref{fig:simRes1betHD} shows the MSE for $\bm{\widehat{\beta}}$. 
As expected, SCADopt resembles very closely the oracle estimator. SCAD2s, which improves upon SCADws, outperforms other feasible estimation methods. TaL performs comparably but has higher biases, and Heur improves upon sparseLTS. 
Lasso breaks down due to the presence of MSOM outliers.
	 	\begin{figure}[ht!]
    		\centering
			\includegraphics[width=0.85\linewidth]{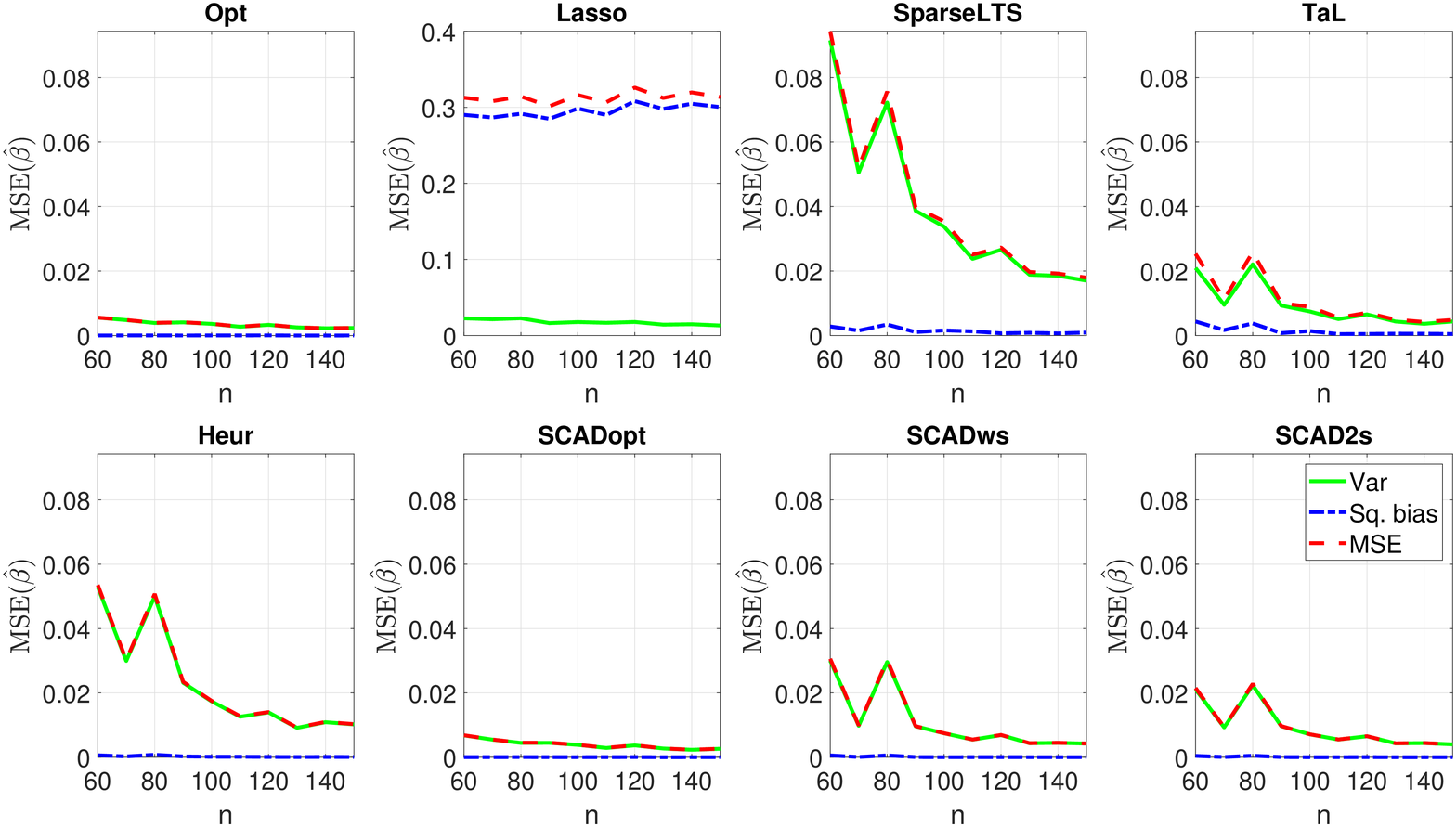}
    		\caption{Scenario 2. MSE$(\widehat{\bm{\beta}})$ comparisons across procedures and sample sizes.}
    		\label{fig:simRes1betHD}
    	\end{figure}
        
The left panels of Figure~\ref{fig:FDpHD} show FPR and FNR for outlier detection.
Unlike the oracle estimator, SCADopt is capable of estimating full weights for VIOMs with negligible residuals (higher FNR), and it is not prone to detecting non-outliers with large residuals by chance (very low FPR).
Notably though, although weights need to be estimated, also SCADws and SCAD2s perform well in both these metrics.
SCAD2s reduces FPR and slightly increases FNR compared to SCADws, which results is an overall performance increase for the iterative approach.
Heur provides larger FPR and smaller FNR. SparseLTS has FPR equal to 0 and large FNR, as it detects only extreme MSOM outliers. TaL performs poorly due to a general down-weighting of all points.

The right panels of Figure~\ref{fig:FDpHD} show FPR and FNR for feature selection. 
SCADopt performs comparably to the oracle estimator.
SCAD2s, which improves upon SCADws, generally outperforms other methods.
TaL produces higher FPR across sample sizes, and Heur provides denser solutions -- but still sparser than sparseLTS.
Lasso performs poorly also here, since it breaks down.
We note that most robust methods are at times affected by MSOMs for smaller sample size (larger FNR and MSE) where their detection is harder.
         \begin{figure}[ht!]
        		\centering
        		\begin{subfigure}{.5\textwidth}
        			\centering
        			\includegraphics[width=1\linewidth]{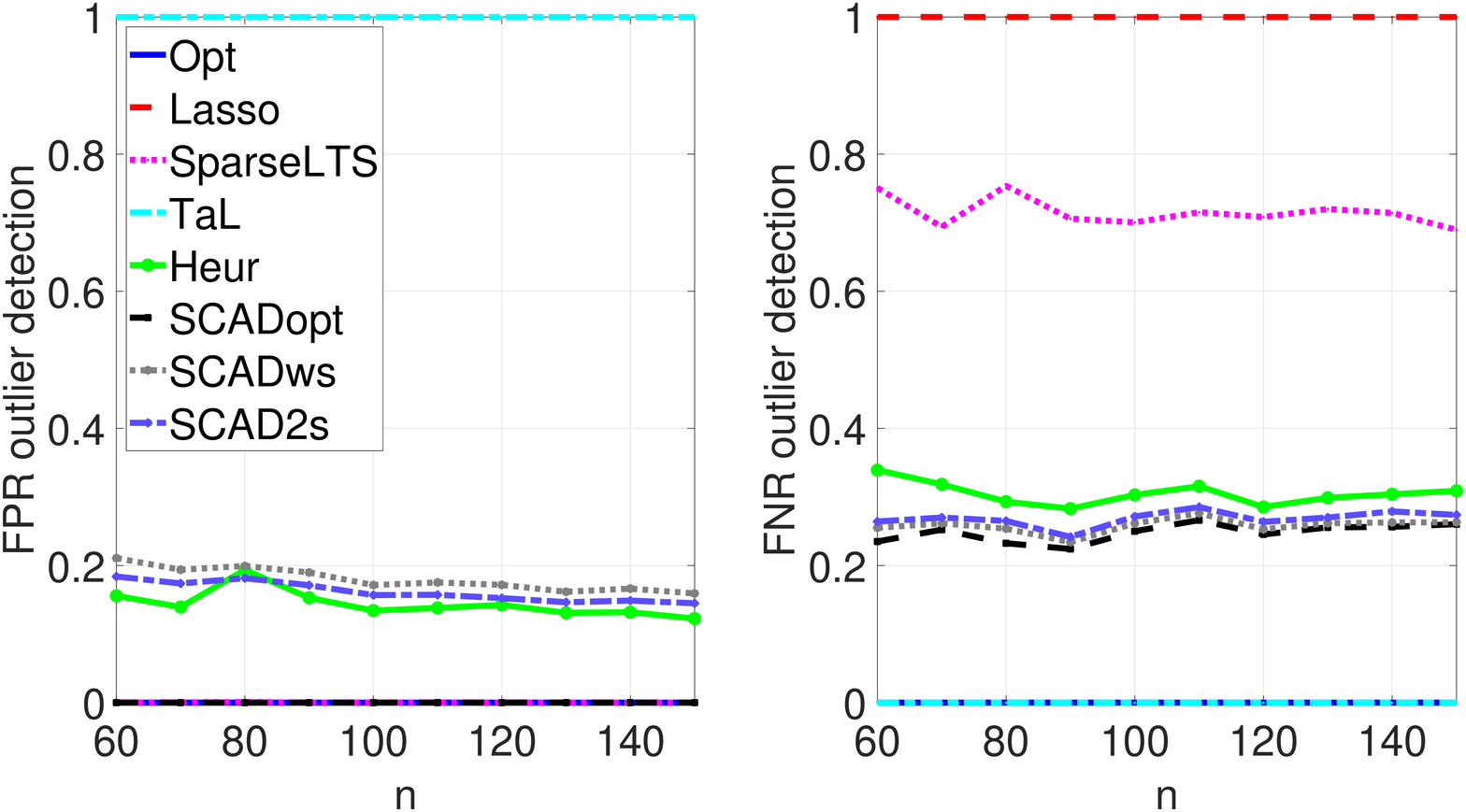}
        		\end{subfigure}%
        		\begin{subfigure}{.5\textwidth}
        			\centering
    			\includegraphics[width=1\linewidth]{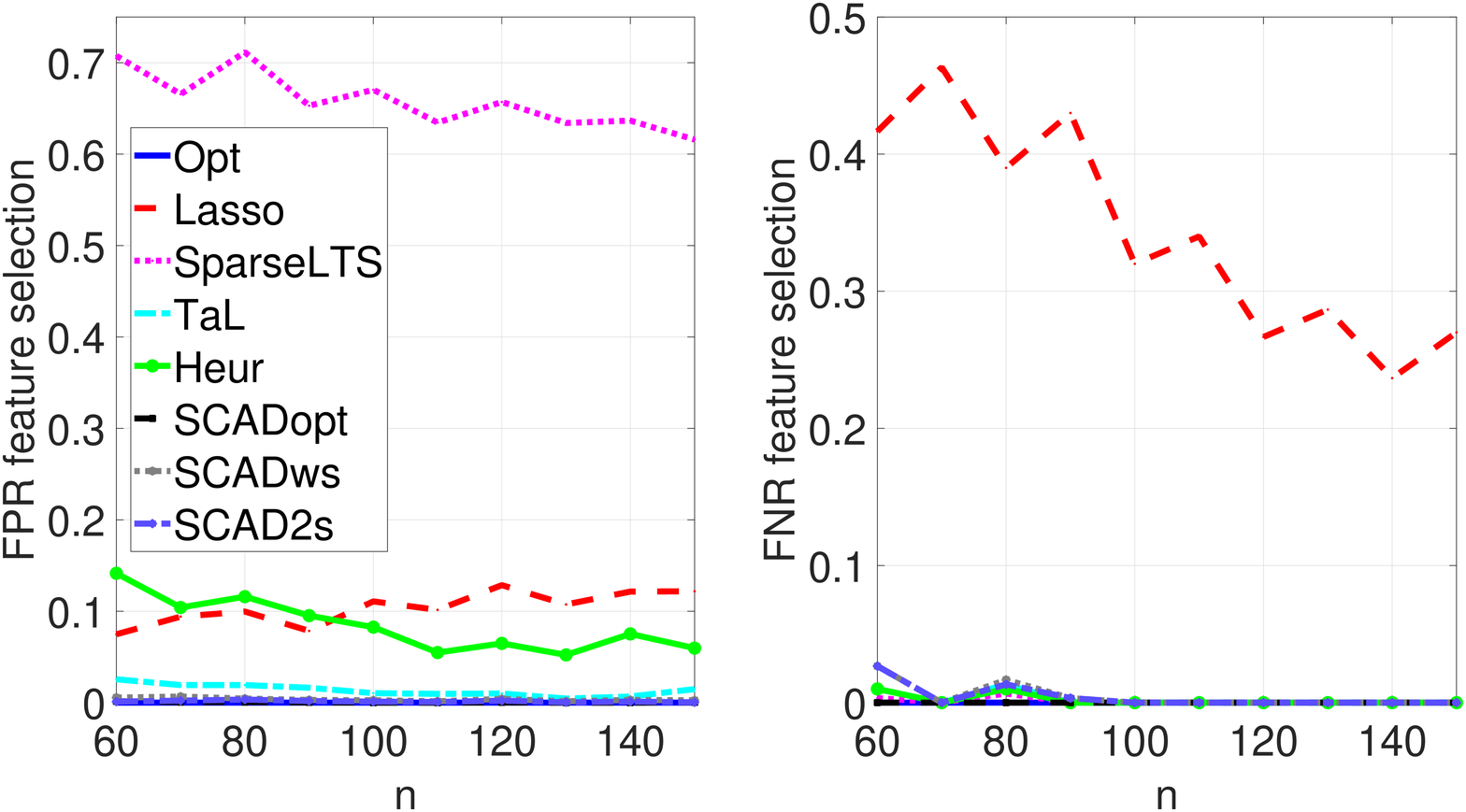}
        		\end{subfigure}
        		\caption{Scenario 2.Comparisons of FPR and FNR for outlier detection (left) and feature selection (right) across procedures and sample sizes.}
        		\label{fig:FDpHD}
        	\end{figure}

\section{An application to the Boston Housing Data} 
\label{sec:appl}

The Boston Housing dataset (\url{http://lib.stat.cmu.edu/datasets/boston}) 
contains $n=506$ housing location and $13$ predictors; 
namely: 
1. {\it crim} (the per capita crime rate), 
2. {\it zn} (the proportion of residential land zoned for lots over 25,000 sq.ft), 
3. {\it indus} (the proportion of non-retail business acres),
4. {\it chas} (a ``Charles River'' dummy), 
5. {\it nox} (the nitrogen oxides concentration in parts per 10 million), 
6. {\it rm} (the average number of rooms per dwelling), 
7. {\it age} (the proportion of owner-occupied units built prior to 1940), 
8. {\it dis} (a weighted mean distance to five Boston employment centers), 
9. {\it rad} (an index of accessibility to radial highways),
10. {\it tax} (the full-value property-tax rate per \$10,000),
11. {\it ptratio} (the pupil-teacher ratio),
12. {\it black} (1000($B_k$ - 0.63)2, where $B_k$ is the proportion of African-American residents), and
13. {\it lstat} (the percentage of the population in lower socioeconomic status).
These are used to explain {\it medv}, the median value of owner-occupied homes in thousand dollars.
            
Using all predictors plus an intercept, we applied the LTS estimator with increasing trimming and computed the robust BIC (see Supplementary Material). 
This helps identify a reasonable trimming level to use across different methods.
The left panel of Figure~\ref{fig:bic} shows that the curve flattens for low levels, with a noticeable drop only for very small amounts of trimming.
With a conservative 10\% trimming, we used SCAD2s to select the relevant features on the full dataset. These are the predictors number 1, 5, 6, 8, 9, 10, 11, 12, 13 (plus the intercept).
The central panel of Figure~\ref{fig:bic} shows the robust BIC recomputed on these features alone. 
There is some evidence of both MSOM outliers (the curve achieves a maximum around 5\% trimming) and VIOM outliers (the curve flattens starting from 15-10\%).
Using again 10\% trimming, the right panel of Figure~\ref{fig:bic} shows the residuals obtained by SCAD2s on the full dataset.
Cases detected as MSOM and VIOM outliers are highlighted.
         \begin{figure}[ht!]
    		\centering
    		\begin{subfigure}{.33\textwidth}
    			\centering
    			\includegraphics[width=1\linewidth]{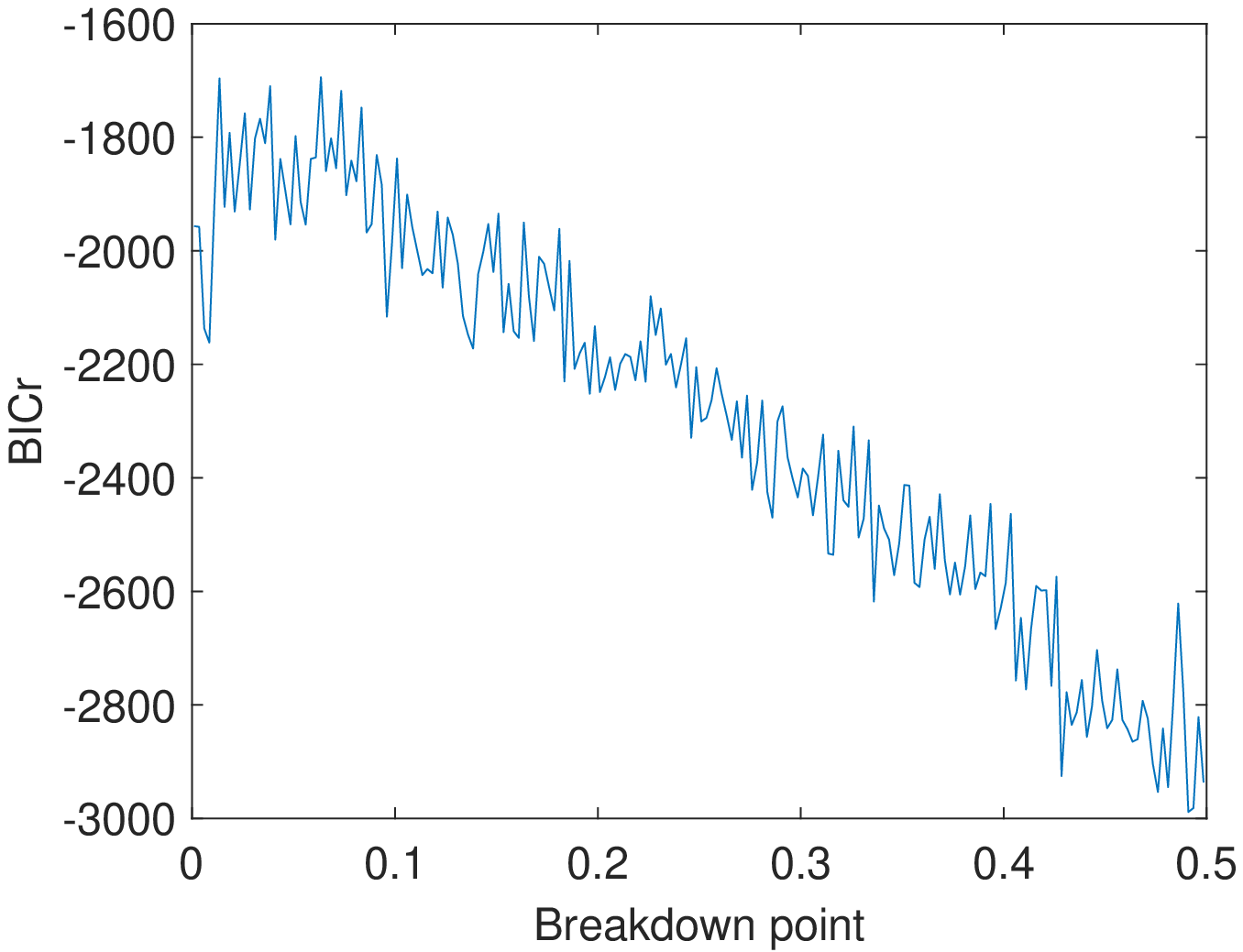}
    		\end{subfigure}%
    		\begin{subfigure}{.33\textwidth}
    			\centering
    			\includegraphics[width=1\linewidth]{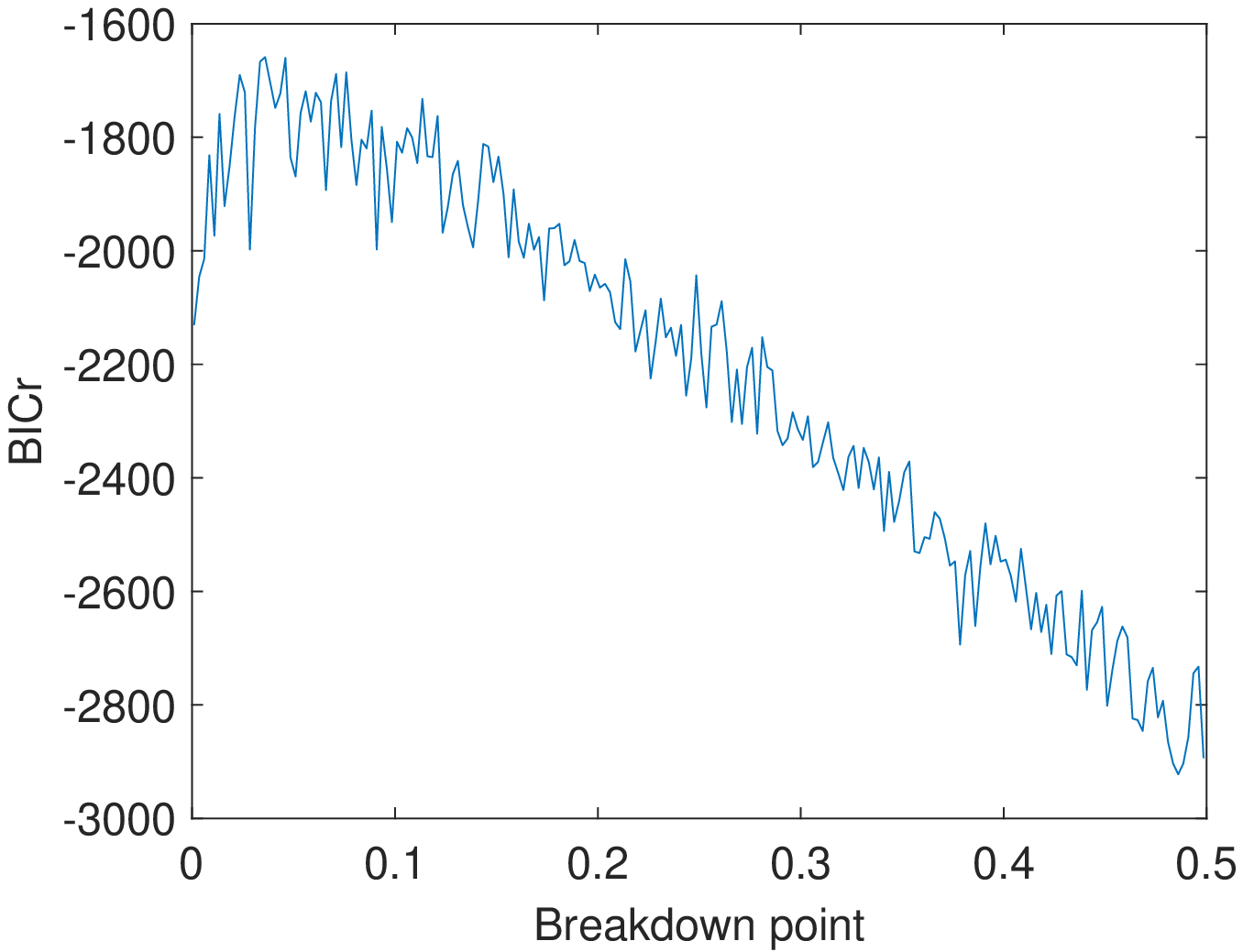}
    		\end{subfigure}%
    		\begin{subfigure}{.33\textwidth}
    			\centering
			\includegraphics[width=1\linewidth]{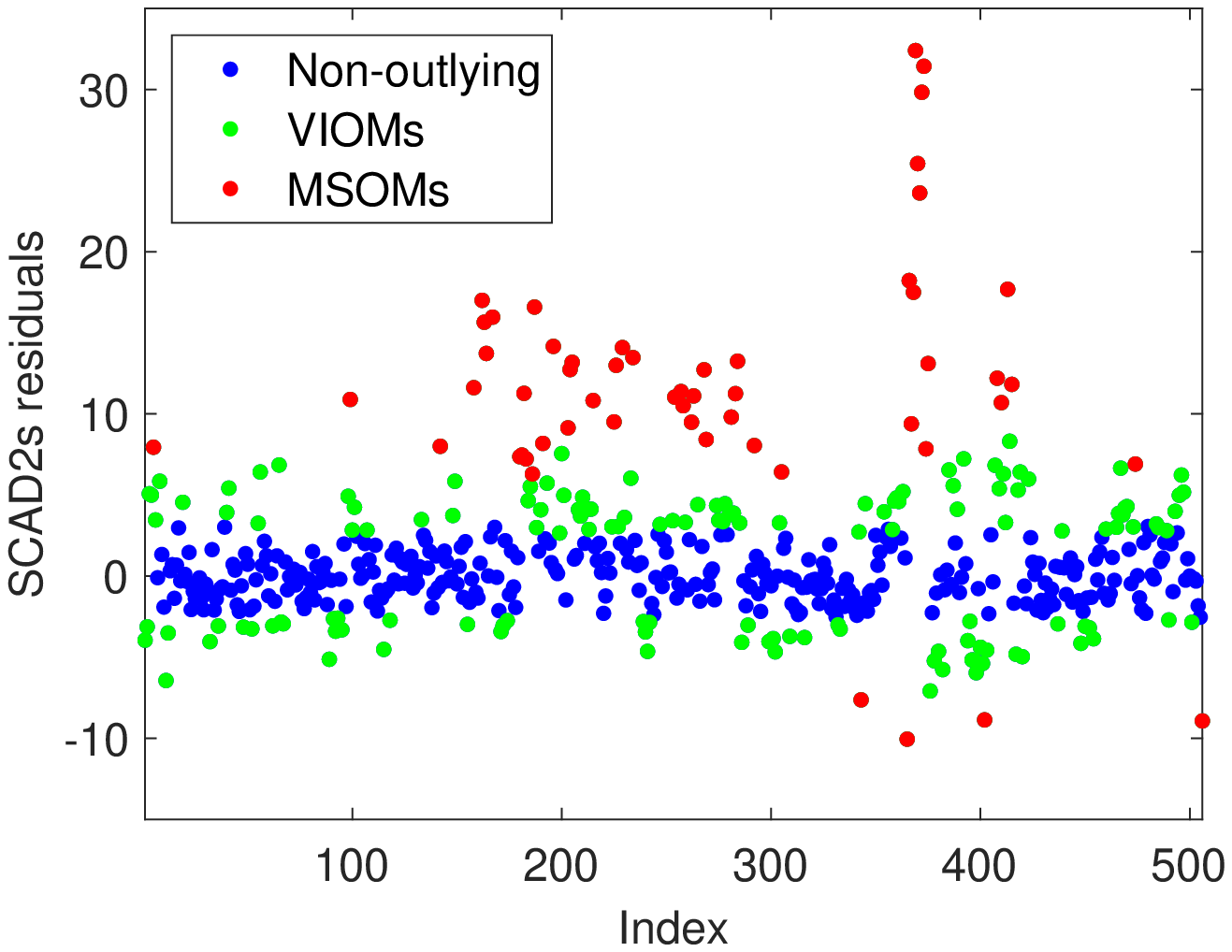}
    		\end{subfigure}
    		\caption{Left: robust BIC computed on all points and features. Center: robust BIC computed on all points and only the features selected using SCAD2s. Right: SCAD2s residuals labeled as non-outlying (blue), MSOM (red), and VIOM (green).}
    		\label{fig:bic}
    	\end{figure}

Next, we extended the analysis along lines similar to \citet{chang2018robust}. We considered 20 random splits of the data in training and testing sets (300 and 206 units, respectively). 
Based on the observations above we used again 10\% trimming across robust methods. 
The left panel of Figure~\ref{fig:appselect} shows box-plots of the sparsity levels, i.e.,~the number of features retained by different methods, across the 20 random training sets.
Some methods do not provide sparse estimates by definition, but also lasso and our heuristic proposal provide very dense solutions.
TaL and sparseLTS provide, respectively, sparser and denser solutions compared to SCAD2s and SCADws. SCAD2s appears to induce slightly more sparsity than SCADws.
The right panel of Figure~\ref{fig:appselect} shows the distribution of the selected features across the 20 random training sets.
The solution for SCAD2s is in line with prior analyses and, unlike TaL, supports the relevance of predictors number 8 and 9 ({\it dis} and {\it rad}).
         \begin{figure}[ht!]
    		\centering
    		\begin{subfigure}{.5\textwidth}
    			\centering
    			\includegraphics[width=1\linewidth]{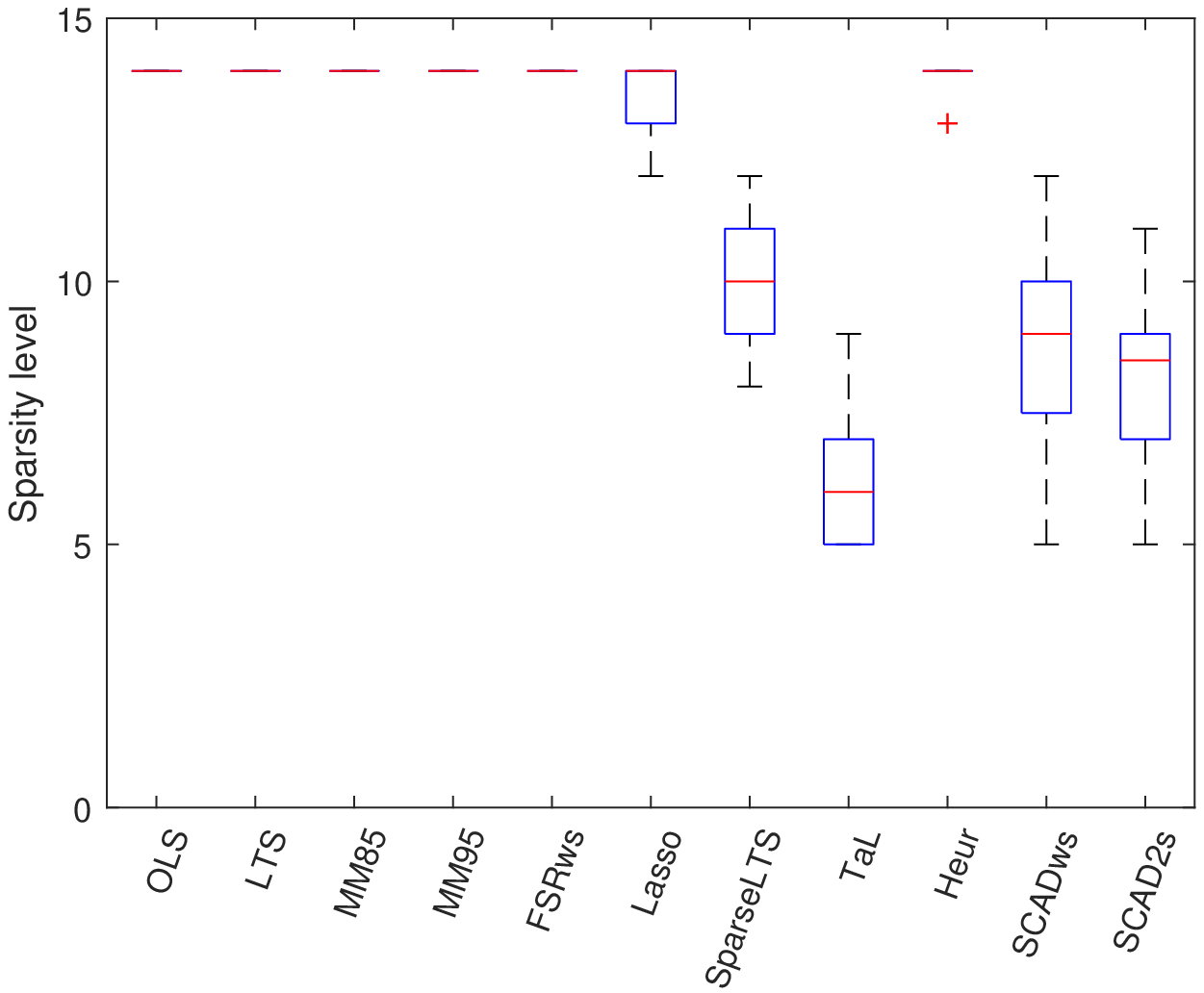} 
    		\end{subfigure}%
    		\begin{subfigure}{.5\textwidth}
    			\centering
			\includegraphics[width=1\linewidth]{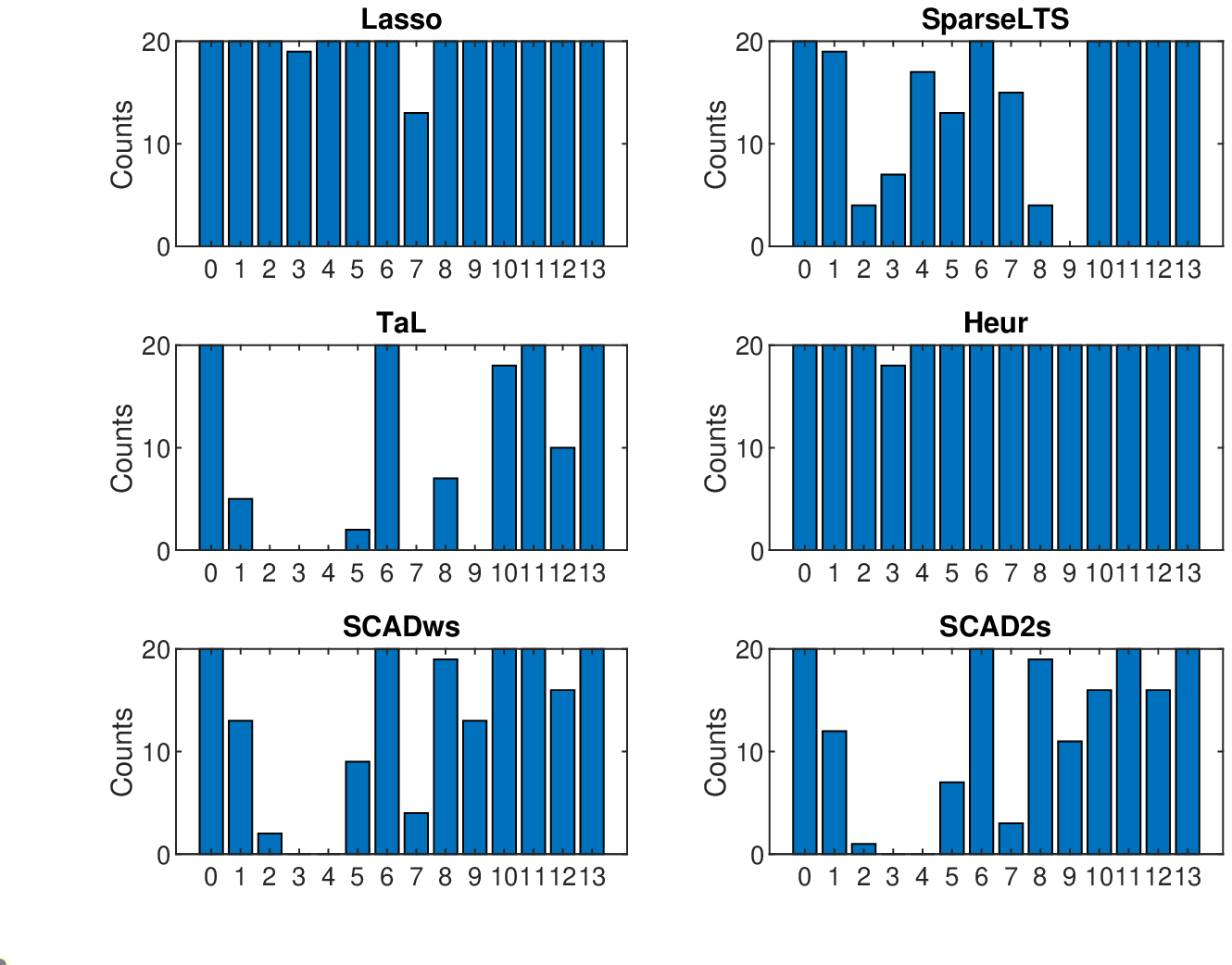}
    		\end{subfigure}
    		\caption{
            Box-plots of the estimated sparsity levels (left) 
            and distribution of the selected features for sparse methods (right) 
            across 20 random training sets for different methods.
    		}
    		\label{fig:appselect}
    	\end{figure}
    	
Figure~\ref{fig:apppred} compares the prediction accuracy of different methods across the 20 random training/testing splits based on the mean absolute (MAPE) and trimmed mean squared (TMSPE) prediction errors, with an upper 10\% trimming.
SCADws and SCAD2s provide a good trade-off between model parsimony and prediction accuracy. 
They outperform TaL (the only method generating sparser solutions) in terms of prediction, independently of the considered quantile.
Our heuristic procedure performs very well -- often better than non-sparse robust estimators -- in terms of prediction, but it has very dense solutions.
         \begin{figure}[ht!]
    		\centering
    		\begin{subfigure}{.5\textwidth}
    			\centering
    			\includegraphics[width=1\linewidth]{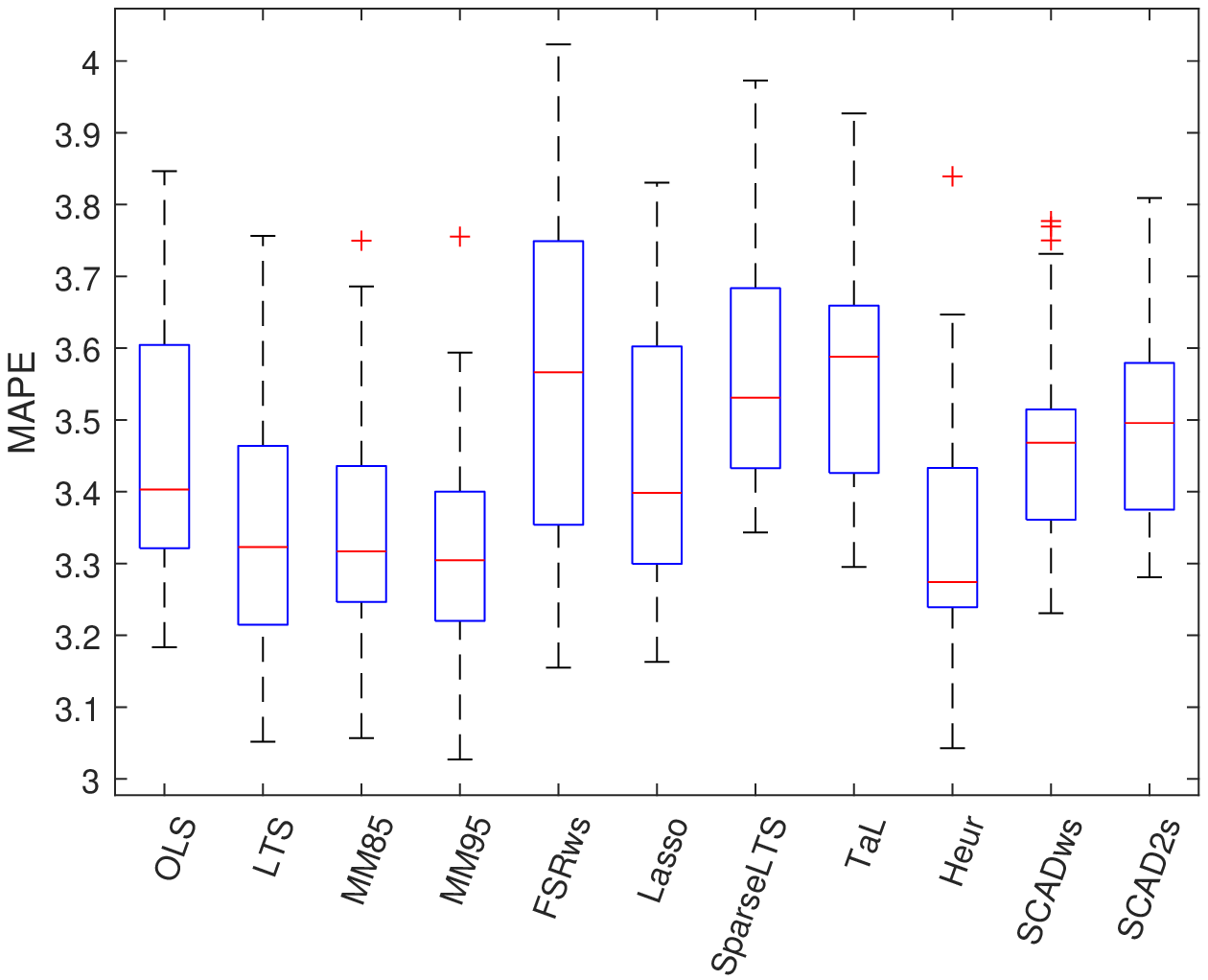}
    		\end{subfigure}%
    		\begin{subfigure}{.5\textwidth}
    			\centering
			\includegraphics[width=1\linewidth]{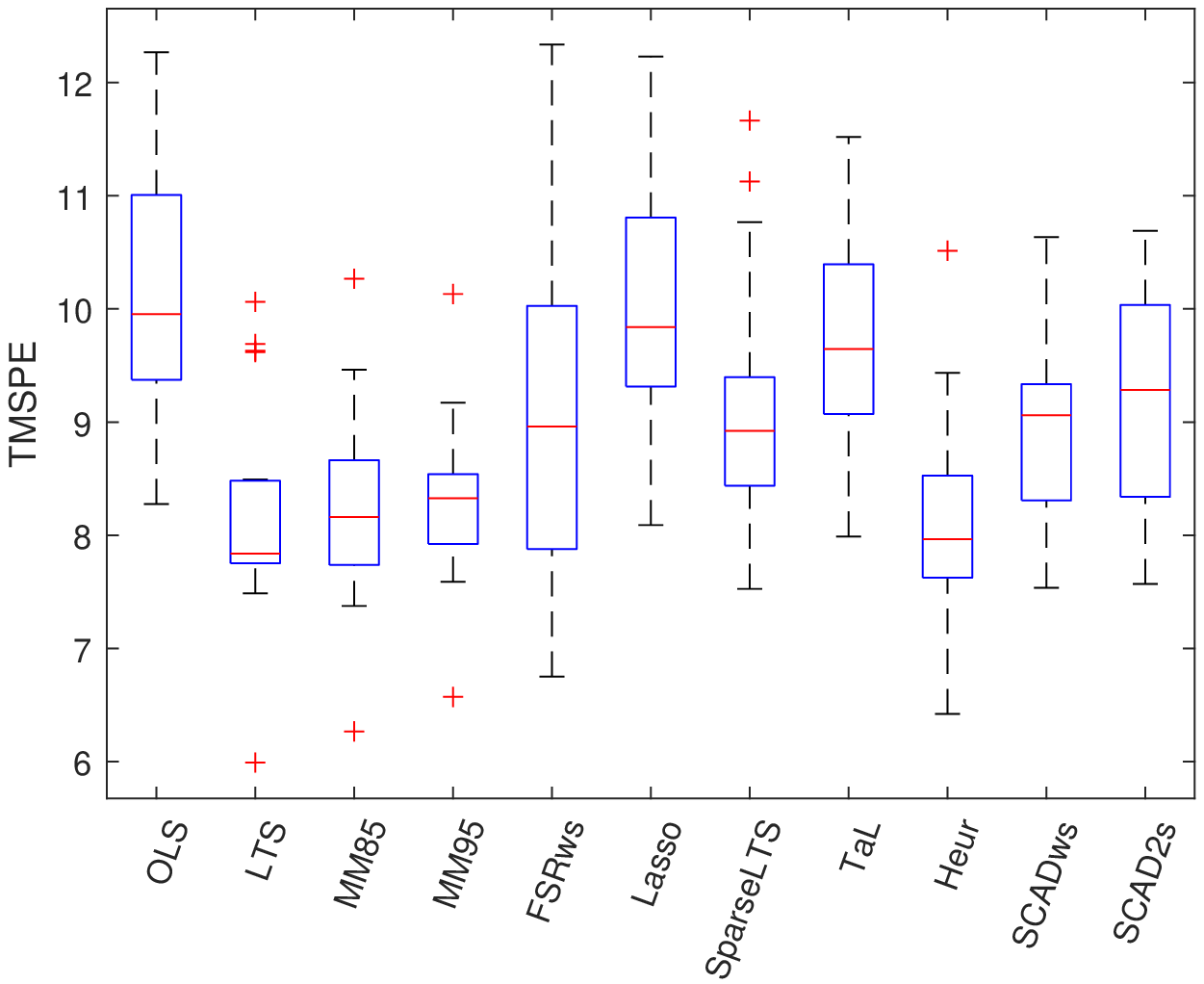}
    		\end{subfigure}
    		\caption{Box-plots of MAPE (left) and TMSPE (right) across 20 random training/testing splits for different methods.}
    		\label{fig:apppred}
    	\end{figure}

\section{Final Remarks}
\label{sec:final}
	
We combine different contamination schemes with sparse estimation methods for linear regression settings. 
This extends robust, sparse estimators based on hard trimming, which explicitly assume only MSOM outliers, to the co-occurrence of VIOM outliers.
Importantly, as we rely on nonconcave penalties, our approach bridges the gap between robust estimation methods enforcing sparsity based on convex penalties, and the use of optimal $L_0$-constraints.
Moreover, unlike methods which provide a general down-weighting for all points based on $M$-estimation, our proposal effectively estimates the weight for each data point. 
Indeed, asymptotically, non-outlying cases receive full weights, MSOMs are excluded from the fit, and only VIOMs are down-weighted.
	    
The theoretical results characterizing our proposal include its high breakdown point, a robust oracle property -- which allows the number of feature to increase exponentially with the sample size -- and the detection of each type of outliers with probability tending to one.
Moreover, including a computationally cheap extra step, our proposal achieves a doubly strong oracle property. This provides optimal units' weights and thus an optimal trade-off between high-breakdown point and efficiency.

Our work can be extended in several directions.
We plan to investigate scenarios with correlated errors, extending our approach for VIOM outlier detection to non-diagonal covariance matrices.
More generally, we are studying high-dimensional mixed-effects linear models affected by data contamination, which allow one to effectively model data with a natural group structure (e.g.,~spatio/temporal relations).
In this setting, VIOM outliers might also arise in the random effects. This has been investigated in \citet{gumedze2010variance} for a single outlier in a known position, but we plan to extend it to the case of multiple MSOM and VIOM outliers in unknown positions.
	    
Moreover, as our theoretical results critically rely on tuning parameters controlling the trade-off between sparsity and efficiency, we are interested in the development of suitable information criteria for sparse models affected by different sources of  contamination, extending the robust BIC introduced in this work.
We are also developing more effective ways to build proxy matrices used in our procedure, as well as iterative approaches.
Finally, we are exploring how to include into our framework cellwise contamination \citep{alqallaf2009propagation}, which is recently receiving a lot of attention for high-dimensional settings.

    \def\spacingset#1{\renewcommand{\baselinestretch}%
    {#1}\small\normalsize} \spacingset{1}

	\bigskip
    \begin{center}
    {\Large\bf SUPPLEMENTARY MATERIAL}
    \end{center}

\spacingset{1}
  
      \phantomsection
	\section*{Appendix A: Theoretical Results}
	\renewcommand{\theequation}{A.\arabic{equation}}
	\setcounter{equation}{0}

       \begin{proof}[Proof of Proposition~\ref{thm:bdp}.]
            For any trimming level $k_n$,
            the objective function in \eqref{eq:reg_fixVIOM_MSOM} subject to integer constraints in 
            \eqref{eq:contrL0} can be equivalently formulated as
            \begin{equation} \label{eq:reforPHI}
                Q( \widehat{\bm{\beta}} ) = 
                \frac{1}{2} 
                \sum_{i=1}^{n - k_n} [ (y_i^*- \bm{\beta}^T \bm{x}_i^* )^2 ]_{i:n}
                +
                (n - k_n) \sum_{j=1}^{p} R_{\lambda}(\lvert \beta_j \lvert) 
            \end{equation}
            where 
            $ ( t_1 )_{1:n} \leq \ldots \leq (  t_n )_{n:n} $ denote the order statistics of $t_i$, 
            $
            \bm{y}^* = \sqrt{\bm{\mathcal{M}}_R} \bm{y}
            $
            and
        	$ \bm{X}^* = \sqrt{\bm{\mathcal{M}}_R} \bm{X} =
        	( \bm{x}_1^*, \ldots, \bm{x}_n^*)^T  
            $. 
            This relies on the fact that a weighted regression of $\bm{y}$ on $\bm{X}$ is equivalent to an unweighted regression of  $\bm{y}^*$ on $\bm{X}^*$, and we also use Proposition~1 in \citet{insolia2020simultaneous} to transform the mean-shift model based on $\bm{\phi}$ to a trimmed loss problem without explicit mean shift parameters.
            Then, denote the contaminated dataset as
            $ \widetilde{\bm{Z}} = 
            [ \widetilde{\bm{y}} , \widetilde{\bm{X}} ] =
            [ ( \bm{y} + \bm{\Delta}_{y} ) , ( \bm{X} + \bm{\Delta}_{X} ) ]
            $.
            We first show that the BdP $ \varepsilon^* \geq ( n - k_n + 1) / n $, and then 
            $ \varepsilon^* \leq ( n - k_n + 1) / n $.
            
            For the first part of the proof assume that 
            $ \widetilde{\bm{Z}} $ contains $ m_M \leq k_n $ outliers.
            Consider 
            $ \widehat{ \bm{\beta} } = 0 $, 
            so that 
           the associated loss
            $$
                Q(\bm{0}) = 
                \sum_{i=1}^{ n -k_n } 
                ( \widetilde{ y }^2_i )_{i:n} 
                \leq \sum_{i=1}^{n-k_n} ( y^2_i )_{i:n} 
                \leq (n - k_n) M_y^2 ,
            $$ 
            where the first inequality relies on the fact that contaminated data might contain inliers (i.e.,~mean shifts can be used to reduce the overall residuals sum of square),
            and 
            $ M_y = \max_{i = 1, \ldots, n}{ \lvert y_i \lvert } $.
            Now consider any other estimate $\widehat{\bm{\beta}}$, and assume that 
            $ \norm{\widehat{\bm{\beta}}}_2 \geq l $ 
            --  
            i.e.,~the estimator might break down -- where
            $ l = \{ ( n- k_n) M_y^2 + 1 \}  / c $ is independent from the contamination mechanism 
            and  $c > 0 $.
            It follows that
            $$
               Q( \widehat{\bm{\beta}} ) \geq
               (n - k_n) \sum_{j=1}^{p} R_{\lambda}(\lvert \beta_j \lvert) \geq
               c ( n- k_n)  \norm{ \bm{\beta} }_2 \geq
               ( n- k_n) M_y^2 + 1  > Q( \bm{0} ) ,
            $$
            where the first inequality immediately follows from \eqref{eq:reforPHI}, and the second inequality is based on the topological equivalence of norms and the definition of SCAD, since 
            $ \norm{ \bm{\beta} }_1 
                \geq \sum_{j=1}^{p} R_{\lambda}(\lvert \beta_j \lvert) 
                \geq c \norm{ \bm{\beta} }_2
            $ for some constant $c>0$ and any $\bm{\beta}$ vector.
            However, 
            $ Q( \widehat{ \bm{\beta} } ) > Q( \bm{0} )  $ leads to a 
            contradiction as the objective function is non-decreasing in the number of non-zero $ \widehat{\beta}_j$ components. 
            Hence, 
            $ \norm{\widehat{\bm{\beta}}}_2 < l $ implies that
            $ \varepsilon^* \geq ( n - k_n + 1) / n $, which concludes the first part of the proof.
            
            For the second part of the proof,
            consider 
            $ m_M > k_n $,
            and assume that
            $ \norm{ \widehat{ \bm{\beta} }(\widetilde{\bm{Z}}) }_2 \leq u $
            (i.e.,~the estimator does not breakdown).
            The objective in \eqref{eq:reforPHI} can be decomposed as
            \begin{align}
                Q( \widehat{\bm{\beta}}) =&
                \sum_{i=1}^{n-m_M} \left[ (\widetilde{y}_i^*- \widehat{\bm{\beta}}^T \widetilde{\bm{x}}_i^* )^2 \right]_{i:n} 
                + \sum_{h=n-m_M+1}^{n-k_n} 
                \left[ 
                (\widetilde{y}_h^*- \widehat{\bm{\beta}}^T \widetilde{\bm{x}}_h^* )^2
                \right]_{h:n}  
                +
                (n - k_n) \sum_{j=1}^{p} R_{\lambda}(\lvert \widehat{\beta}_j \lvert) 
                \nonumber \\
                \geq& 
                \left[ 
                \{ ( y_i^*- \bm{\beta}^T \bm{x}_i^* )
                + 
                (\Delta_{y_i} - \widehat{\bm{\beta}}^T \bm{\Delta}_{x_i} )  
                \}^2
                \right]_{i=n-m_M+1}
                + 
                (n - k_n) \sum_{j=1}^{p} R_{\lambda}(\lvert \widehat{\beta}_j \lvert) 
                \label{eq:proof_bdp}
            \end{align}
            since at least one of the $m_M$ outliers might be included in the fit -- i.e.,~the $(n-n_0+1)$-th ordered squared residual if contamination is adversarial.
            Hence, since mean shifts $ \Delta_{y_i} $ and $ \bm{\Delta}_{x_i} $ can take arbitrary values, it is easy to see that \eqref{eq:proof_bdp} is unbounded similarly to OLS. 
            This contradicts $ \norm{ \widehat{ \bm{\beta} }(\widetilde{\bm{Z}}) }_2 \leq u $ and proves the result.
        \end{proof}
        
        \begin{proof}[Proof of Theorem~\ref{thm:fix}.]
            It extends Theorem~1 in \citet{fan2012variable} to the presence of MSOM contamination.
            Specifically, we can use the same argument, but their conditions must hold at least on $n-m_M$ (uncontaminated) points as opposed to  $n$. 
            Since $k_n$ largest residuals (say, $k_n = m_M$) are always discarded from our loss in \eqref{eq:reg_fixVIOM_MSOM}, we thus need to ensure that these trimmed points encompass MSOM outliers.
            Condition~\ref{cond2}(D) guarantees this, similarly to Theorem 3 in \citet{insolia2020simultaneous}, so that MSOM outliers have largest residuals for any model of size $k_p \leq p_0$.
            See \citet{fan2012variable} for details.
        \end{proof}

        \begin{proof}[Proof of Theorem~\ref{thm:rnd}]
            This result immediately follows from Theorem~2 in \citet{fan2012variable} specifically focusing on VIOM
            outliers as random effects (i.e.,~our term $\bm{I}_n \bm{\gamma}$ instead of $ \bm{Z} \bm{b}$).
            However, in \citet{fan2012variable} the dimension of the random effects $ \bm{b}$ can increase exponentially with the sample size $n$, but in our formulation $ \bm{\gamma} $ can only increase linearly with $n-k_n$.
            Thus, our conditions in list~\ref{cond2} might be relaxed to account only for VIOMs.
            Nevertheless, these more general conditions allow one to extend our results also to the presence of additional (pure) random effects, whose size can increase exponentially with $n-k_n$. 
        \end{proof}

            \begin{proof}[Proof of Theorem~\ref{thm:opt_bdp_eff}(1).]
                The proofs for Theorem~\ref{thm:opt_bdp_eff} follow some lines of the argument in Theorems~1 and 3 of \citet{liu2013asymptotic}, where an OLS or ridge fit is computed on top of the features selected by lasso.
                
                Here with a slight abuse of notation, we denote 
                 $ P(\mathcal{S}) = P( \widehat{\mathcal{S}} = \mathcal{S} )$ and
                 $ P( \widetilde{\mathcal{S}}) = P( \widehat{\mathcal{S}} \neq \mathcal{S} )$,
                 where $\widehat{\mathcal{S}} = \{ \widehat{\mathcal{S}}_\beta , \widehat{\mathcal{S}}_\phi, \widehat{\mathcal{S}}_\gamma  \} $.
                 Furthermore, we indicate 
                  as
                $         
                \widehat{\bm{\beta}} \lvert \mathcal{ \widehat{S} }
                $
                 the estimated coefficients conditionally on the selected model,
                which is abbreviated as
                $ 
                \widehat{\bm{\beta}}_{\mathcal{ \widehat{S}}}
                $.
                It is also assumed that, conditioned on any selected model 
                $\mathcal{ \widehat{S}}$, units weights $ \widehat{\bm{W}} $ 
                are deterministic.
                
                By the law of total expectations and using $ \| a +b \|^2 \leq 2( \| a \|^2 + \| b \|^2 )$, it follows that
                \begin{align}
                    \| E \widehat{\bm{\beta}} - \bm{\beta}_0  \|_2^2 
                    &= \| E \widehat{\bm{\beta}}_{\mathcal{S}} P(\mathcal{S} )
                    + E \widehat{\bm{\beta}}_{\widetilde{\mathcal{S}}} P( \widetilde{\mathcal{S}}) 
                    - \bm{\beta}_0
                    \|_2^2  \nonumber \\
                    &\leq 
                    2 \| E \widehat{\bm{\beta}}_{\mathcal{S}} P(\mathcal{S} )
                    - \bm{\beta}_0
                    \|_2^2 
                    + 
                    2 \| E \widehat{\bm{\beta}}_{\widetilde{\mathcal{S}}} P( \widetilde{\mathcal{S}})  \|_2^2 \nonumber \\
                    &=  
                    2\| E \{ (\bm{X}_{\mathcal{S}}^T \widehat{\bm{W}} \bm{X}_{\mathcal{S}} )^+ 
                    \bm{X}_{\mathcal{S}}^T \widehat{\bm{W}} \bm{y}  \} P(\mathcal{S} )
                    - \bm{\beta}_0
                    \|_2^2 
                    + 
                    2 P( \widetilde{\mathcal{S}})  \| E \widehat{\bm{\beta}}_{\widetilde{\mathcal{S}}}  \|_2^2  \nonumber \\
                    &=  
                    2\|  P(\mathcal{S} ) \bm{\beta}_0 
                    - \bm{\beta}_0
                    \|_2^2 
                    + 
                    2 P( \widetilde{\mathcal{S}})  \| E \widehat{\bm{\beta}}_{\widetilde{\mathcal{S}}}  \|_2^2  \nonumber \\
                    &=  
                    2\|   \bm{\beta}_0 ( P(\mathcal{S} ) 
                    - 1 )
                    \|_2^2 
                    + 
                    2 P( \widetilde{\mathcal{S}})  \| E \widehat{\bm{\beta}}_{\widetilde{\mathcal{S}}}  \|_2^2  \nonumber \\
                    &= 
                    2 P( \widetilde{\mathcal{S}})  \{ \|  \bm{\beta}_0
                    \|_2^2 
                    + 
                    \| E \widehat{\bm{\beta}}_{\widetilde{\mathcal{S}}}  \|_2^2 \}  \label{pr:1pt1_1} .
                \end{align}
                Further, using Jensen's inequality and the fact that $ \| A b\| \leq \| A \| \|b \| $ provides
                \begin{align}
                    \| E \widehat{\bm{\beta}}_{\widetilde{\mathcal{S}}}  \|_2^2 &\leq
                     E \|  (\bm{X}_{\widetilde{\mathcal{S}}}^T \widehat{\bm{W}} \bm{X}_{\widetilde{\mathcal{S}}} )^+ 
                    \bm{X}_{\widetilde{\mathcal{S}}}^T \widehat{\bm{W}} \bm{y} 
                    \|_2^2 
                    \nonumber \\
                    &\leq
                     E \|  (\bm{X}_{\widetilde{\mathcal{S}}}^T \widehat{\bm{W}} \bm{X}_{\widetilde{\mathcal{S}}} )^+ 
                    \bm{X}_{\widetilde{\mathcal{S}}}^T \widehat{\bm{W}}^{1/2} \|_2^2 
                    \| \widehat{\bm{W}}^{1/2} \bm{y} \|_2^2 
                    \nonumber \\
                    &=
                     \Lambda_{\max}  \{ (\bm{X}_{\widetilde{\mathcal{S}}}^T \widehat{\bm{W}} \bm{X}_{\widetilde{\mathcal{S}}} )^+ \}
                    E \| \widehat{\bm{W}}^{1/2} \bm{X} \bm{\beta}_0 + \widehat{\bm{W}}^{1/2} \bm{\varepsilon} \|_2^2
                    \nonumber \\
                    &= 
                      \Lambda_{\max}  \{ (\bm{X}_{\widetilde{\mathcal{S}}}^T \widehat{\bm{W}} \bm{X}_{\widetilde{\mathcal{S}}} )^+ \}
                    E (\| \widehat{\bm{W}}^{1/2} \bm{X} \bm{\beta}_0 \|_2^2 + \varepsilon^T \widehat{\bm{W}} \varepsilon )  \nonumber \\
                    &= 
                      \Lambda_{\max}  \{ (\bm{X}_{\widetilde{\mathcal{S}}}^T \widehat{\bm{W}} \bm{X}_{\widetilde{\mathcal{S}}} )^+ \}
                    (\| \widehat{\bm{W}}^{1/2} \bm{X} \bm{\beta}_0 \|_2^2 + \text{tr}(\widehat{\bm{W}} ) \sigma^2  ) 
                    \label{pr:1pt1_2}    \\
                    &\leq 
                      \Lambda_{\max}  \{ (\bm{X}_{\widetilde{\mathcal{S}}}^T \widehat{\bm{W}} \bm{X}_{\widetilde{\mathcal{S}}} )^+ \}
                    (\|  \bm{X} \bm{\beta}_0 \|_2^2 + n \sigma^2  ) 
                    ,
                    \nonumber
                \end{align}
                where $\Lambda_{\max}(\cdot)$ denotes the largest eigenvalue, and for a real matrix $A$, the spectral norm $ \norm{A}_2 = \sqrt{\Lambda_{\max}(A A^T )} = \sqrt{\Lambda_{\max}(A^T A )}$.
                In our case,
                \begin{align}
                \| 
                (\bm{X}_{\widetilde{\mathcal{S}}}^T \widehat{\bm{W}} \bm{X}_{\widetilde{\mathcal{S}}} )^+ 
                    \bm{X}_{\widetilde{\mathcal{S}}}^T \widehat{\bm{W}}^{1/2} \|_2^2 
                    &=
                    \Lambda_{\max} \{
                    (\bm{X}_{\widetilde{\mathcal{S}}}^T \widehat{\bm{W}} \bm{X}_{\widetilde{\mathcal{S}}} )^+ 
                    \bm{X}_{\widetilde{\mathcal{S}}}^T \widehat{\bm{W}}
                    \bm{X}_{\widetilde{\mathcal{S}}}
                    (\bm{X}_{\widetilde{\mathcal{S}}}^T \widehat{\bm{W}} \bm{X}_{\widetilde{\mathcal{S}}} )^+
                    \} \nonumber \\ 
                    &=  \Lambda_{\max}  \{ (\bm{X}_{\widetilde{\mathcal{S}}}^T \widehat{\bm{W}} \bm{X}_{\widetilde{\mathcal{S}}} )^+ \}, \nonumber
                \end{align}
                where the last equality follows from the property of a generalized inverse
                $A^+ A A^+ = A^+$.
                Combining \eqref{pr:1pt1_1} and \eqref{pr:1pt1_2}
                leads to the desired results.
            \end{proof}
                   
        \begin{proof}[Proof of Theorem~\ref{thm:opt_bdp_eff}(2).]
            
            Introducing the WLS oracle estimator $\widehat{\bm{\beta}}_0$
            and using the 
            fact that 
            $$
            E \| \widehat{\bm{\beta}}_0 - \bm{\beta}_0  \|_2 =
            E \| (\bm{X}_{\mathcal{S}}^T \widehat{\bm{W}} \bm{X}_{\mathcal{S}} )^+ 
                    \bm{X}_{\mathcal{S}}^T \widehat{\bm{W}} \bm{\varepsilon}   \|_2 = 0 
            $$
            provides
            \begin{align}
                E \|  \widehat{\bm{\beta}} - \bm{\beta}_0  \|_2^2 
                &=  E \|  \widehat{\bm{\beta}} + \widehat{\bm{\beta}}_0 - \widehat{\bm{\beta}}_0 - \bm{\beta}_0  \|_2^2 \nonumber \\  
                &= 
                E \|  \widehat{\bm{\beta}} - \widehat{\bm{\beta}}_0\|_2^2 + E \| \widehat{\bm{\beta}}_0 - \bm{\beta}_0  \|_2^2 
                \nonumber \\
                &=
                E \|  \widehat{\bm{\beta}} - \widehat{\bm{\beta}}_0\|_2^2 
                +
                \sigma^2 \text{tr}(\bm{\Sigma}_{X}^{-1} ) / \text{tr}(\widehat{\bm{W}})  \label{pr:1pt1_3}
            \end{align}
            the last equality follows from the MSE for the WLS oracle estimator and such term cannot be improved.
            Thus, we control the first term as follows
            \begin{align}
                E \|  \widehat{\bm{\beta}} - \widehat{\bm{\beta}}_0\|_2^2 &=
                E \|  \widehat{\bm{\beta}}_{\mathcal{S}} - \widehat{\bm{\beta}}_0\|_2^2 P(\mathcal{S} ) +
                E \|  \widehat{\bm{\beta}}_{\widetilde{\mathcal{S}}} - \widehat{\bm{\beta}}_0\|_2^2 P( \widetilde{\mathcal{S}}) \nonumber \\
                &= E \|  \widehat{\bm{\beta}}_{\widetilde{\mathcal{S}}} - \widehat{\bm{\beta}}_0\|_2^2 P( \widetilde{\mathcal{S}}) , \label{pr:1pt1_4}
            \end{align}
            where the first equality relies on the law of total expectations and the last one uses the fact that 
            $ \widehat{\bm{\beta}}_{\widehat{\mathcal{S}}} = \widehat{\bm{\beta}}_0 $ conditioned on $ \{ \widehat{\mathcal{S}} = \mathcal{S} \} $.
            
            Further, note that
            \begin{align}
                E \|  \widehat{\bm{\beta}}_{\widetilde{\mathcal{S}}} - \widehat{\bm{\beta}}_0\|_2^2 
                &\leq 
                2\{ E \| \widehat{\bm{\beta}}_{\widetilde{\mathcal{S}}} \|_2^2 + E \| \widehat{\bm{\beta}}_0 \|_2^2 \}
                \nonumber \\
                &=
                2 \{
                E \| (\bm{X}_{\widetilde{\mathcal{S}}}^T \widehat{\bm{W}} \bm{X}_{\widetilde{\mathcal{S}}} )^+ 
                    \bm{X}_{\widetilde{\mathcal{S}}}^T \widehat{\bm{W}} \bm{y} \|_2^2 + 
                E \| (\bm{X}_{\mathcal{S}}^T \widehat{\bm{W}} \bm{X}_{\mathcal{S}} )^+ 
                    \bm{X}_{\mathcal{S}}^T \widehat{\bm{W}} 
                    \bm{y} \|_2^2 
                \}
                \nonumber \\
                &\leq 
                2 E \|  \widehat{\bm{W}}^{1/2} \bm{y}  \|_2^2
                \left[
                 \Lambda_{\max} \{ (\bm{X}_{\widetilde{\mathcal{S}}}^T \widehat{\bm{W}} \bm{X}_{\widetilde{\mathcal{S}}} )^+ \} +
                 \Lambda_{\max} \{ (\bm{X}_{\mathcal{S}}^T \widehat{\bm{W}} \bm{X}_{\mathcal{S}} )^+ \}  \right], 
                \label{pr:1pt1_5}
            \end{align}
            where the first upper bound follows from $ \| a +b \|^2 \leq 2( \| a \|^2 + \| b \|^2 )$, and the second one uses $ \| A b\| \leq \| A \| \|b \| $.
            Finally, combining
            $$
                E \|  \widehat{\bm{W}}^{1/2} \bm{y}  \|_2^2
                \leq 
                E ( \| \widehat{\bm{W}}^{1/2} \bm{X} \bm{\beta}_0 \|_2^2 + \varepsilon^T \widehat{\bm{W}} \varepsilon )
                = 
                \| \widehat{\bm{W}}^{1/2} \bm{X} \bm{\beta}_0 \|_2^2 + \text{tr}(\widehat{\bm{W}} ) \sigma^2  
                \leq
                \|  \bm{X} \bm{\beta}_0 \|_2^2 + n \sigma^2 
            $$
            with \eqref{pr:1pt1_3}, \eqref{pr:1pt1_4}, and \eqref{pr:1pt1_5}
            concludes the proof.
        \end{proof}
        
        \begin{proof}[Proof of Theorem~\ref{thm:opt_bdp_eff}(3).]
            Under the conditions in lists~\ref{cond1}-\ref{cond3},
            as $ (n - m_M) \to \infty $, it follows that
             $ P (\widehat{\mathcal{S}} = \mathcal{S} ) \to 1 $ for some suitable constants.
            Thus, $ \widehat{\bm{\beta}} $ has an asymptotic normal distribution as it is a linear combination of normal distributions
            \begin{align}
             \widehat{\bm{\beta}} &= (\bm{X}_{\mathcal{S}}^T \widehat{\bm{W}} \bm{X}_{\mathcal{S}} )^{-1} 
                    \bm{X}_{\mathcal{S}}^T \widehat{\bm{W}} \bm{y}  \nonumber \\
                    & = 
                    (\bm{X}_{\mathcal{S}}^T \widehat{\bm{W}} \bm{X}_{\mathcal{S}} )^{-1} 
                    \bm{X}_{\mathcal{S}}^T \widehat{\bm{W}}
                    (\bm{X}_{\mathcal{S}} \bm{\beta}_0 +  \bm{\varepsilon}) \nonumber \\
                    &= \bm{\beta}_0 + 
                    (\bm{X}_{\mathcal{S}}^T \widehat{\bm{W}} \bm{X}_{\mathcal{S}} )^{-1} 
                    \bm{X}_{\mathcal{S}}^T \widehat{\bm{W}} \bm{\varepsilon}
                    \sim N(\bm{\beta}_0, 
                    \sigma^2
                    (\bm{X}_{\mathcal{S}}^T \widehat{\bm{W}} \bm{X}_{\mathcal{S}} )^{-1}  ) ,
                    \nonumber
            \end{align}
            and 
            $ \widehat{\bm{W}} = \bm{V}^{-1} $ 
            guarantees that it asymptotically reaches maximum efficiency.
        \end{proof}

    \phantomsection
	\section*{Appendix B: Technical Details}
	\label{sec:appendix}
	\renewcommand{\theequation}{B.\arabic{equation}}
	\setcounter{equation}{0}

       \subsection*{B.1 Choice of the Proxy Matrix $\mathcal{M}$}
       \label{sec:app_choiceM}
       
          For mixed-effects linear models without data contamination as in Section~\ref{secsub:mixed_model}, \citet{fan2012variable} propose to replace $ \sigma^{-2} \bm{ \mathcal{B}} $ in \eqref{eq:Pgamma} with a proxy matrix $ \bm{\mathcal{M}}_b$.
          They show that under mild conditions it is safe to choose $\bm{\mathcal{M}}_b = \log(n) \bm{I}_n$, as the eigenvalues of $ \bm{Z}^T \bm{P}_x \bm{Z}$ and $\bm{Z} \bm{Z}^T$ have magnitude increasing with $n$, so that they are likely to dominate the eigenvalues of $\bm{\mathcal{M}}_b$ for a large enough $n$.
          While this choice excludes cross-correlations in the random effects, it avoids the estimation of a large number of parameters as in the case of an unstructured covariance matrix. 
        
        In our formulation the terms $\bm{\mathcal{M}}_R$ and $\bm{\mathcal{M}}_\gamma$ in \eqref{eq:reg_fixVIOM_MSOM} and \eqref{eq:reg_randVIOM_MSOM} are proxies for the unknown $\bm{P}_R$ and $\bm{\Gamma}$, respectively.
        Following \citet{fan2012variable}, in our implementation we use 
        $\bm{\mathcal{M}}_R = \bm{\mathcal{M}}_\gamma  = \log(n) \bm{I}_n$ on the first iteration. 
        If the 3-step procedure is re-iterated, such as in SCAD2s, we use estimated weights $ \widehat{\bm{ W}}$ from the previous iteration for their update.

       \subsection*{B.2 Weights Estimation}
      \label{sec:app_weights}
      
           The formulation in \eqref{eq:reg_randVIOM_MSOM} highlights that if $\widehat{\gamma}_i = 0 $ also the corresponding variance inflation $\widehat{\omega}_i = 0$.
            However, it might be of interest to estimate $\omega_i$ when the corresponding $\widehat{\gamma}_i \neq 0 $. 
            A similar reasoning holds for step~3 of the heuristic method described in Section~\ref{secsub:proposal_heuristic}.
            Note that
            \begin{equation}
                w_i=v_i^{-1}=(1+\omega_i)^{-1} =  
                ( 1 + \text{var}(\gamma_i) / \sigma^2 )^{-1},  \nonumber
            \end{equation}
            which can be estimated as follows:
            \begin{enumerate}
            
                \item Apply REMLE assuming that the units corresponding by non-zero components in $\widehat{ \bm{\gamma} }$ arise from a VIOM.
                In principle, all weights should be jointly estimated, although this can be computationally heavy for large problems.
                A similar approach was used by \citet[p.2060]{fan2012variable} in one of their examples.
                Similarly to \citet{insolia2020ViomMsom}, we also consider single-weights estimates as in FSRws, where each VIOM outlier is separately included in the model and estimated.
                This is the approach used in our simulations and application.
                
                \item The quantity $ \gamma_j^2 /n $ can be used as an estimate of $ \text{var}(\gamma_j)$ \citep[p.~2053 Eq.~20]{fan2012variable}.
                Thus, one can consider
                $
                    w_i = (1 + \widehat{\gamma}_i^2 c_1 /  \widehat{\sigma}^2)^{-1}
                $
                where $c_1$ is a normalizing constant and the value $c_1=1/n$ was suggested by the authors.
                
                \item One can treat the selected random effects $\gamma_i$ as additional fixed effects and apply a ridge penalty \citep{hoerl1970aridge}.
                This can be considered optimal and is motivated by the fact that assuming a normal prior $N(\bm{0}, \sigma^2 \bm{\Gamma})$ on $\bm{\gamma}$  leads to the ridge estimator as the maximum posterior probability estimator. 
                Indeed, the estimates $\widehat{\bm{\gamma}}$ represent prediction residuals, so that their shrinkage performs a down-weighting scheme.
                Moreover, \citet{grandvalet1998least} showed that adaptive ridge is equivalent to lasso estimation; this can be useful to simultaneously select and estimate optimal units' weights (e.g.,~combining Steps 2 and 3 of our main proposal and/or  heuristic procedure).
            \end{enumerate}

        \phantomsection
    	\subsection*{B.3 Parameter Tuning } 
        \renewcommand{\thetable}{B.\arabic{table}}
    	\setcounter{table}{0}
            
            For feature selection and MSOM detection  is essential to tune the sparsity level and the amount of trimming.
            We propose to combine the approach in \citet{insolia2020simultaneous}
            and \citet{Riani2021bicw}.
            Specifically, in low-dimensional models affected by MSOM contamination, 
            \citet{Riani2021bicw}
            introduced the following robust version of BIC to tune the trimming level for hard-trimming estimators:
            \begin{equation}
                \operatorname{BICW}=
                -n 
                    \log \left\{ R\left(\widehat{\beta}_{h}\right) \bigg/ \left( \sigma^{2}_h h \right) \right\}
                -\left\{
                    p + k_n
                \right\} 
                \log n , \nonumber
	        \end{equation}
	        where 
	        $h=n-k_n$ and
	        $ R (\widehat{\beta}_{h} ) $ is the residual sum of square based on the $h$ observations contributing to the loss.
	        The associated variance of the truncated normal distribution containing a central portion $h/n$ of the full distribution is
	        \begin{equation}
	            \sigma^{2}(h)=
	            1-\frac{2 n}{h} \Phi^{-1}\left(\frac{n+h}{2 n}\right) \phi\left\{\Phi^{-1}\left(\frac{n+h}{2 n}\right)
	            \right\} , \nonumber
	        \end{equation}
	        where $\phi(\cdot)$ and $\Phi(\cdot)$ denote the probability cumulative density function for the Gaussian distribution, respectively.
	        
	        In our work, we consider the following extension of BICW for high-dimensional settings:
	        \begin{equation}
	         \operatorname{BICr}=
                  -n 
                \left[ 
                    \log \left\{ R\left(\widehat{\beta}_{h}\right) \bigg/ \left( \sigma^{2}_h  h \right) \right\}
                \right]
                -\left\{
                    k_p + k_n
                \right\} 
                \log n , \nonumber
	        \end{equation}
	        where $k_p = \lvert \widehat{\mathcal{S}}_{\beta} \rvert  $ denotes the sparsity level for feature selection.
	        This formulation improves and extends the robust BIC proposed in \citet{insolia2020simultaneous}. 
	        In principle one should consider a range of trimming values $k_n$ and shrinkage parameter $\lambda$ (the latter determines $k_p$). However, to reduce the computational burden, we often fix one of the two parameter and tune only the other. 
            Moreover, to take into account the co-occurrence of VIOM outliers this might be generalized further, similarly to the CAIC and extended CAIC discussed in Section~\ref{secsub:mixed_model}.
	        
        \phantomsection
    	\subsection*{B.4 Parallel Between our Heuristic Approach and  $M$-estimation} 
    	\renewcommand{\thetable}{B.\arabic{table}}
    	\setcounter{table}{0}
    	\renewcommand{\thefigure}{B.\arabic{figure}}
    	\setcounter{figure}{0}

                    The proposed heuristic method has a parallel with the following multi-stage, penalized $M$-estimation procedure.
                    
                    Step 1 is equivalent to an adaptive hard-trimming, sparse estimator (i.e.,~it selects features and assigns binary weights) and guarantees an high-breakdown point.
                    This step aims to exclude MSOMs and select only the relevant features 
                    (see for instance \citealt{alfons2013sparse, kurnaz2017robust, insolia2020simultaneous}).
                    Step 2 corresponds to an adaptive ``truncated'' $M$-estimator, where only the most extreme cases are down-weighted.
                    In full generality, this estimator takes the form
                  $
                   \widehat{\bm{\beta}}
                   = \arg \min_\beta  
                    \sum_{i=1}^{n}  \rho( \bm{e} / \sigma )  .
                $
                    Here the idea is that the $n-m_M-m_V$
                    uncontaminated points receive full weights as in OLS, but only VIOMs are down-weighted according to the $\rho(\cdot)$ function in use, and MSOMs (if present) are excluded from the fit.
                 \begin{figure}[ht!]
            		\centering
            		\begin{subfigure}{.33\textwidth}
            			\centering
            			\includegraphics[width=1\linewidth]{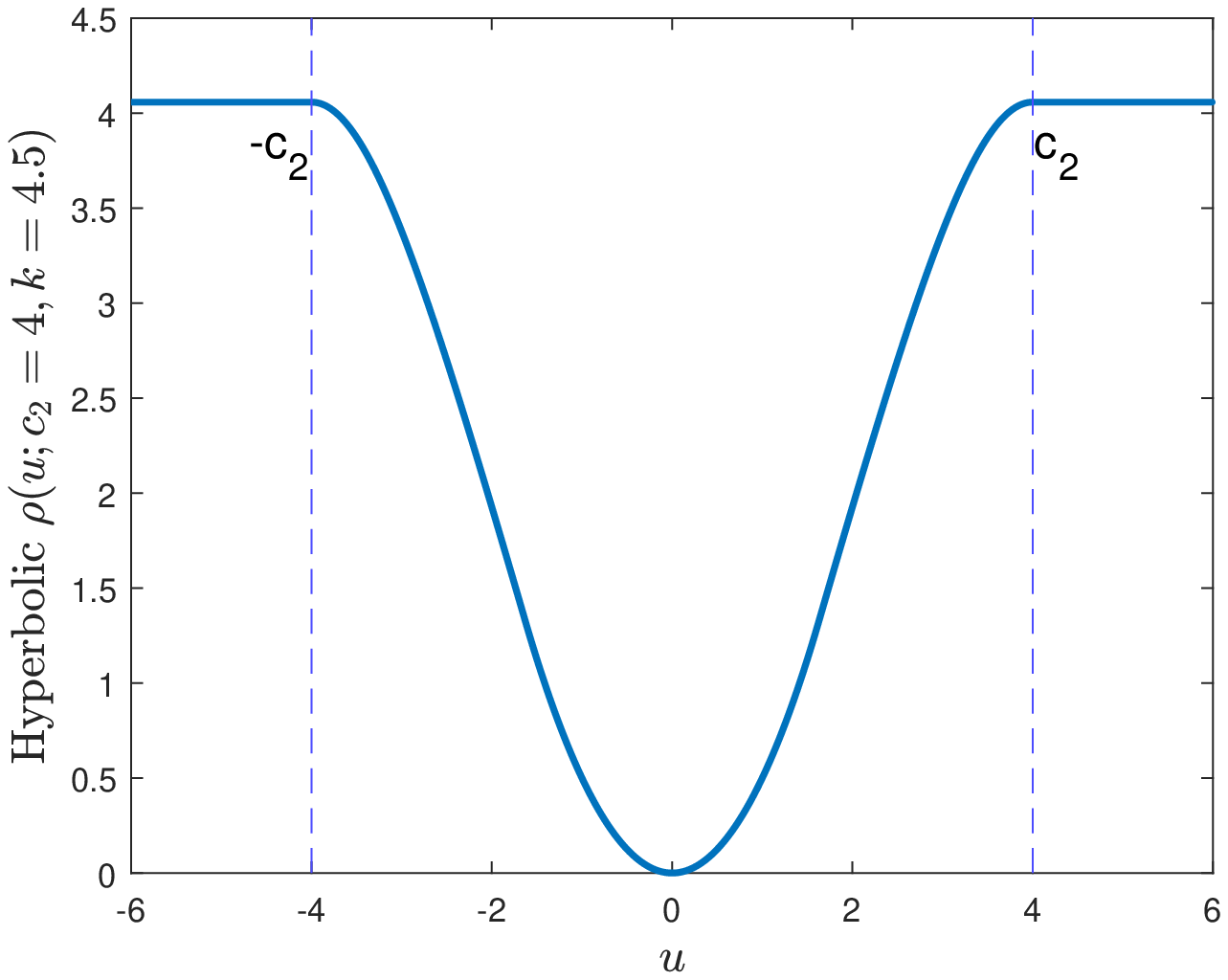}
            		\end{subfigure}%
            		\begin{subfigure}{.33\textwidth}
            			\centering
            			\includegraphics[width=1\linewidth]{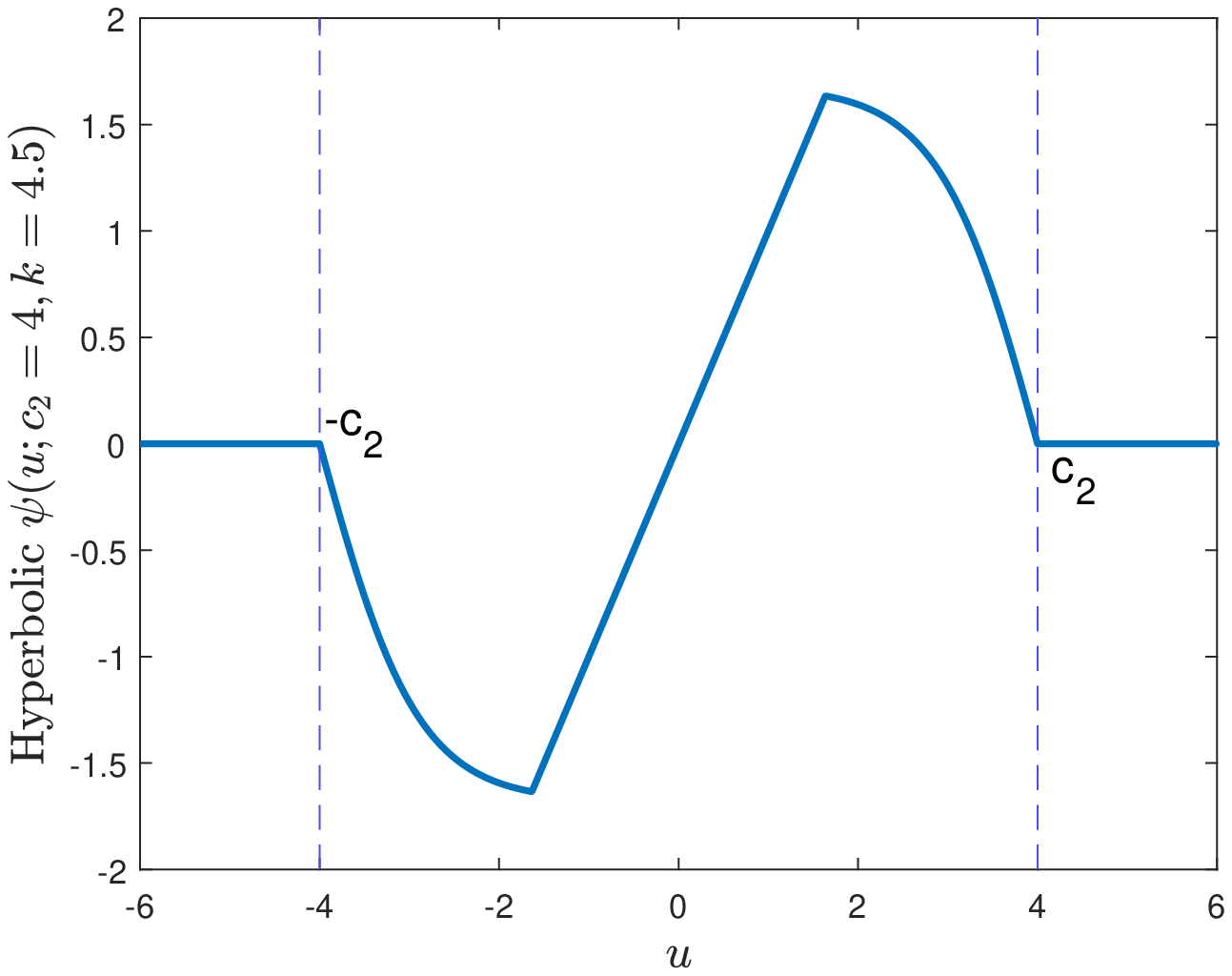}
            		\end{subfigure}%
            		\begin{subfigure}{.33\textwidth}
            			\centering
        			\includegraphics[width=1\linewidth]{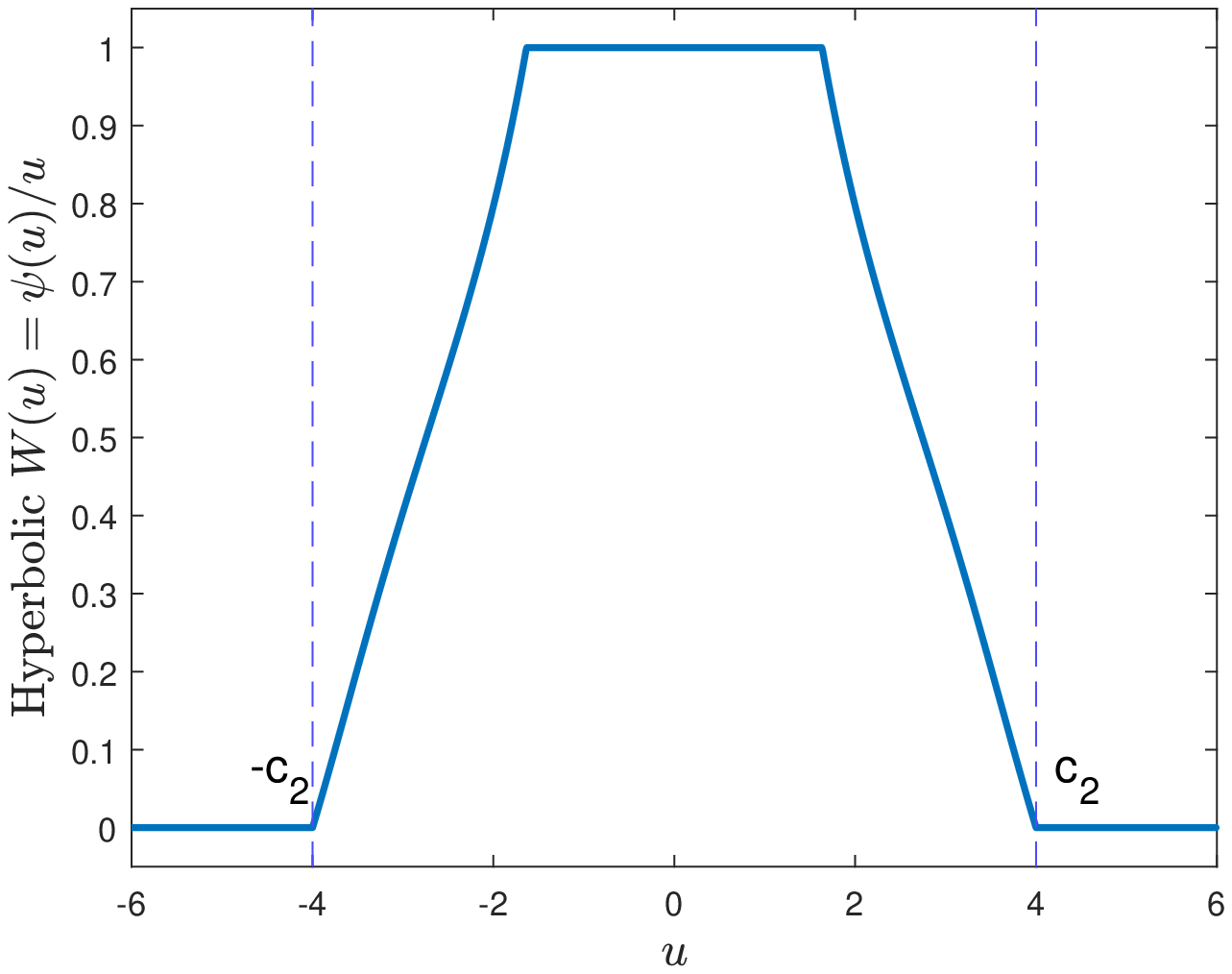}
            		\end{subfigure}
            		\caption{Hyperbolic Tangent $\rho$ function (left panel), $\psi$ function (central panel), and weight function (right panel) for $ c_2 = 4$ and $k=4.5$.}
            		\label{fig:tanh}
            	\end{figure}
                	For instance, this has a parallel with the \textit{hyperbolic tangent} $\rho(\cdot)$ function, which can be considered as refinement of Hampel's piecewise linear redescending function
                	and is related to the \textit{change of variance curve} \citep{hampel1981change}.
                    Tanh-estimators are more easily defined in terms of their derivatives,
                    and the corresponding $\psi(\cdot)$  function 
                    is
                    $$
                    \psi(u)=\left\{\begin{array}{ll}
                    u & \text{ if }|u| \leq c_{1} \\
                    \{ A(k-1) \}^{1 / 2} \tanh \left[ \frac{1}{2}\left \{ (k-1) B^{2} / A\right \}^{1 / 2} 
                    (c_2 - | u |) \right] \operatorname{sign}( u )
                     & \text { if } c_1 \leq|u| \leq c_2 \\
                    0 & \text { if }|u|>c_2
                    \end{array}\right.
                    $$
                    for suitable constants $k$, $A$, $B$, $c_1$, and $c_2$, where $0<c_1<c_2$ satisfies
                    $$
                    c_1= \{ A(k-1)\}^{1 / 2} \tanh \left[\frac{1}{2}\left\{(k-1) B^{2} / A\right\}^{1 / 2}(c_2-c_1)\right] .
                    $$
                    These constants are traditionally  computed iteratively, based on the Newton-Raphson algorithm and numerical integration.
                    Figure~\ref{fig:tanh} shows the corresponding $\rho$, $\psi$, and weight functions for $ c_2 = 4$ and $k=4.5$.
                    
                    Unlike tanh-estimators, our heuristic proposal does not pre-specify a trade-off between breakdown point and efficiency, but this is adaptively tuned as follows.
                    The rejection point $c_2$ approximately corresponds to the smallest standardized residual for the MSOMs detected at step 1.
                    Similarly, the constant $c_1$ is set to the value of the largest standardized residual for points which are not affected by MSOM or VIOM.
                    Specifically, for our heuristic proposal, $c_1$ and $c_2$ can be computed based on order statistics from the scaled residuals obtained at step 1.
                    Ideally, assuming without loss of generality that all outliers have sizeable residuals, these corresponds to the 
                    $(n-m_V-m_M)$-th and $(n-m_M)$-th order statistics of the absolute standardized residuals, respectively.

    \phantomsection
	\section*{Appendix C: Code}

	Our code is available upon request.

    \bibliography{biblio.bib}

\end{document}